\documentclass[12pt]{article}

\usepackage{amsmath,amssymb}
\usepackage{ifthen}
\usepackage{url}
\usepackage{tikz,color}
\usetikzlibrary{shapes,snakes,decorations}
\tikzstyle{every picture}=[thick,inner sep=0mm]
\tikzstyle{vertex}=[minimum size=1mm,circle,fill]

\usepackage[normalem]{ulem}

\newcommand{\deriv}[2]{\textrm{d}#1/\textrm{d}#2}

\numberwithin{equation}{section}

\def\nfrac#1#2{{\textstyle\frac{#1}{#2}}}

\usepackage{tikz}
\usetikzlibrary{shapes,snakes}

\tikzstyle{every picture}=[thick,inner sep=0mm]
\tikzstyle{vertex}=[minimum size=2mm,circle,fill]

\usepackage{amsthm}
\newtheorem{thm}{Theorem}[section]
\newtheorem{lem}[thm]{Lemma}

\newtheorem{lemma}[thm]{Lemma}

\usepackage[colorlinks=true,citecolor=black,linkcolor=black,urlcolor=blue]{hyperref}
\newcommand{\arxiv}[1]{\href{http://arxiv.org/abs/#1}{\texttt{arXiv:#1}}}

\newcommand*\dtv{\ensuremath{\mathrm{d_\mathrm{TV}}}}
\newcommand*\M{\ensuremath{\mathcal{M}}}
\newcommand\mix[1][]{\ensuremath{\tau\ifthenelse{\equal{#1}{}}{}{_{\mathrm{#1}}}(\varepsilon)}}
\newcommand*\set[1]{\ensuremath{\left\{#1\right\}}}

\newcommand*{\function}[2]{\ensuremath{\mathrm{#1}\!\left(#2\right)}}	
\newcommand*{\bigO}[1]{\function{O}{#1}}
\newcommand*{\bigTh}[1]{\function{\Theta}{#1}}	

\newcommand*\h{\ensuremath{h}}  
\newcommand*\mf{\ensuremath{\mathcal{M}_\mathrm{F}}}
\newcommand*\ms{\ensuremath{\mathcal{M}_\mathrm{S}}}
\newcommand*\OS{\ensuremath{\Omega_\mathrm{S}}}
\newcommand*\OF{\ensuremath{\Omega_\mathrm{F}}}
\newcommand*\GS{\ensuremath{\Gamma_\mathrm{S}}}
\newcommand*\GF{\ensuremath{\Gamma_\mathrm{F}}}
\newcommand*\msc{\ensuremath{\mathcal{M}_\mathrm{SC}}}
\newcommand*\GSC{\ensuremath{\Gamma_\mathrm{SC}}}
\newcommand*\SSC{\ensuremath{\Sigma_\mathrm{SC}}}
\newcommand*\SF{\ensuremath{\Sigma_\mathrm{F}}}

\newcommand*\piSC{\ensuremath{\pi_\mathrm{SC}}}
\newcommand*\piF{\ensuremath{\pi_\mathrm{F}}}
\newcommand*\PS{\ensuremath{P_\mathrm{S}}}
\newcommand*\PSC{\ensuremath{P_\mathrm{SC}}}
\newcommand*\PF{\ensuremath{P_\mathrm{F}}}
\newcommand*\Switch[5][]{\ensuremath{(#2#3,#4#5\ifthenelse{\equal{#1}{}}{\Rightarrow}{\Rightarrow_{#1}}#2#4,#3#5)}}

\newcommand*\flippath{\ensuremath{\sigma}}
\newcommand*\DS{\ensuremath{D_\mathrm{SC}}}
\newcommand*\DF{\ensuremath{D_\mathrm{F}}}

\newcommand{\st}{\rule{0pt}{8pt}}	

\title{The flip Markov chain\\ for connected regular graphs~\footnote{An earlier version of this paper appeared as an extended abstract in PODC 2009~\cite{CDH}.}~\footnote{This research was partly performed  while the first three authors were visiting the Simons Institute for the Theory of Computing.}}

\author{
Colin Cooper\thanks{Supported by EPSRC Research Grant EP/M004953/1.}\\
{\small Department of Computer Science}\\[-0.8ex]
{\small King's College London}\\[-0.8ex]
{\small London WC2R 2LS, U.K.}\\
{\small \texttt{colin.cooper@kcl.ac.uk}}\\
\and
Martin Dyer\footnotemark[3]\\
{\small School of Computing}\\[-0.8ex]
{\small University of Leeds}\\[-0.8ex]
{\small Leeds LS2 9JT, U.K.}\\
{\small \texttt{m.e.dyer@leeds.ac.uk}}\\
\and
Catherine Greenhill\thanks{Supported by the Australian Research Council Discovery Project DP140101519.}\\
{\small School of Mathematics and Statistics}\\[-0.8ex]
{\small UNSW Australia}\\[-0.8ex]
{\small Sydney, NSW 2052, Australia}\\
{\small \texttt{c.greenhill@unsw.edu.au}}\\
\and
Andrew Handley\footnote{Supported by EPSRC Research Grant EP/D00232X/1.}\\
{\small School of Computing}\\[-0.8ex]
{\small University of Leeds}\\[-0.8ex]
{\small Leeds LS2 9JT, U.K.}\\
{\small \texttt{jobriath@gmail.com}}
}

\date{13 June 2018}

\begin{document}

\maketitle

\begin{abstract}
Mahlmann and Schindelhauer (2005) defined a Markov chain which
they called $k$-Flipper, and showed that it is irreducible on
the set of all connected regular graphs of a given degree (at least 3).
We study the 1-Flipper chain, which we call the flip chain,
and prove that the flip chain converges rapidly to the uniform
distribution over connected $2r$-regular graphs with $n$ vertices, 
where $n\geq 8$ and $r = r(n)\geq 2$.
Formally, we prove that the distribution of the flip chain 
will be within $\varepsilon$ of uniform in total variation
distance after $\text{poly}(n,r,\log(\varepsilon^{-1}))$ steps.
This polynomial upper bound on the mixing time is given explicitly,
and improves markedly
on a previous bound given by Feder et al.\ (2006).
We achieve this improvement by using a direct two-stage
canonical path construction, which we define in a general setting.

This work has applications to decentralised networks based on random regular
connected graphs of even degree, as a self-stabilising protocol in which nodes
spontaneously perform random flips in order to repair the network.

\bigskip
\noindent \emph{Keywords:}\ Markov chain; graph; connected graph; regular graph
\end{abstract}

\section{Introduction}

Markov chains which walk on the set of random regular graphs
have been well studied. The switch chain
is a very natural Markov chain for sampling random regular graphs.
A transition of the switch chain is called a \emph{switch},
in which two edges are deleted and replaced with two other edges,
without changing the degree of any vertex. The switch chain was
studied by Kannan et al.~\cite{kannan99markovbipartite} for bipartite graphs, and
in~\cite{cooper05sampling} for regular graphs.
Although a switch changes only a constant number of edges per transition,
it is not a local operation since these edges may be anywhere in the
graph.

Random regular graphs have several properties which make them
good candidates for communications networks: with high probability,
a random $\Delta$-regular graph has logarithmic
diameter and high connectivity, if $\Delta\geq 3$. (See,
for example,~\cite{Wormald99}.)
Bourassa and Holt~\cite{bourassa03swan} proposed a protocol for a
decentralised communications network based on random regular graphs
of even degree. Since then, Markov chains which sample regular graphs
using local operations have been applied to give
\emph{self-stabilising} or \emph{healing} protocols for
decentralised communications networks which have a regular topology.
To help the network to
recover from degradation in performance due to arrivals and departures,
clients in a decentralised network can spontaneously perform local
transformations to re-randomize the network, thereby recovering
desirable properties.

Mahlmann and Schindelhauer~\cite{schindelhauer05kflipper} observed
that the switch operation is unsuitable for this purpose, since a badly-chosen
switch may disconnect a connected network. Furthermore, since the
switch move is non-local there are
implementation issues involved in choosing two random edges in a decentralised
network. (If the network is an expander then this problem can be solved
by performing two short random walks in the network and taking the final edge
of each for the next operation~\cite{bourassa03swan}.)

To overcome these problems, Mahlmann and Schindelhauer proposed
an alternative operation which they called a $k$-Flipper. When $k=1$,
the 1-Flipper (which we call the \emph{flip}) is simply a restricted
switch operation, in which the two edges to be switched must be at
distance one apart in the graph. Figure~\ref{fig:switchAndFlip} illustrates
a switch and a flip.
\medskip

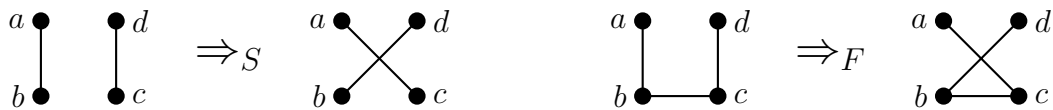
\begin{figure}[ht!]
\begin{center}
\begin{tikzpicture}
\draw [fill] (-1,0) circle (0.1);
\node [left] at (-1.2,0) {$b$};
\draw [fill] (-1,1) circle (0.1);
\node [left] at (-1.2,1.0) {$a$};
\draw [fill] (0,0) circle (0.1);
\node [right] at (0.2,0) {$c$};
\draw [fill] (0,1) circle (0.1);
\node [right] at (0.2,1) {$d$};
\draw [-] (-1,0) -- (-1,1);
\draw [-] (0,0) -- (0,1);
\node at (1.5,0.5) {{\Large $\Rightarrow_S$}};
\draw [fill] (3,0) circle (0.1);
\node [left] at (2.8,0) {$b$};
\draw [fill] (3,1) circle (0.1);
\node [left] at (2.8,1.0) {$a$};
\draw [fill] (4,0) circle (0.1);
\node [right] at (4.2,0) {$c$};
\draw [fill] (4,1) circle (0.1);
\node [right] at (4.2,1) {$d$};
\draw [-] (3,1) -- (4,0);
\draw [-] (3,0) -- (4,1);
\draw [fill] (7,0) circle (0.1);
\node [left] at (6.8,0) {$b$};
\draw [fill] (7,1) circle (0.1);
\node [left] at (6.8,1.0) {$a$};
\draw [fill] (8,0) circle (0.1);
\node [right] at (8.2,0) {$c$};
\draw [fill] (8,1) circle (0.1);
\node [right] at (8.2,1) {$d$};
\draw [-] (7,1) -- (7,0) -- (8,0) -- (8,1);
\node at (9.5,0.5) {{\Large $\Rightarrow_F$}};
\draw [fill] (11,0) circle (0.1);
\node [left] at (10.8,0) {$b$};
\draw [fill] (11,1) circle (0.1);
\node [left] at (10.8,1.0) {$a$};
\draw [fill] (12,0) circle (0.1);
\node [right] at (12.2,0) {$c$};
\draw [fill] (12,1) circle (0.1);
\node [right] at (12.2,1) {$d$};
\draw [-] (11,1) -- (12,0) -- (11,0) -- (12,1);
\end{tikzpicture}
\caption{Left: a switch operation, denoted by $\Rightarrow_S$. Right: a
flip, denoted by $\Rightarrow_F$.}
\label{fig:switchAndFlip}
\end{center}
\end{figure}

By design, the flip operation preserves the degree of all vertices
and preserves connectivity. Indeed, it is the smallest such randomising
operation: no edge exchange on fewer than four vertices
preserves the degree distribution, and the only shallower tree than the 3-path
is the 3-star, which has no degree-preserving edge exchange.
In order to use flips as a self-stabilising protocol, each peer waits for
a random time period using a Poisson clock, after which it instigates a new flip.
This involves communication only with peers at distance at most two away
in the network.

We will show that flips have asymptotic behaviour
comparable to switches, in that both operations randomise a network in
polynomial time.  Simulations also suggest that they operate equally fast, up
to a constant factor (see~\cite[Section 2.3.6]{AHthesis}).
A discussion of a distributed implementation of the flip operation is
given in~\cite{CDH}.

Our approach will be to relate the mixing time of the flip chain $\mf$
to an intermediate chain $\msc$, which is the switch chain restricted to connected
graphs.
In turn, the mixing time of the connected switch chain $\msc$
will be related
to the mixing time of the (standard) switch chain $\ms$, using the
analysis from~\cite{cooper05sampling}. Each of these two steps will be
performed using a technique which we call ``two-stage direct canonical
path construction'', described in Section~\ref{sec:two-stage-direct}.
Specifically, we prove the following result.
This is a corrected version of the bound given
in~\cite{CDH}, which failed on certain rare graphs.

\begin{thm}
For each $n\geq 8$ let $\Delta=\Delta(n)\geq 4$ be a positive even integer
such that $n\geq \Delta +1$.
The mixing time of the flip Markov chain $\mf$ on the set of $\Delta$-regular graphs with
$n$ vertices is at most
\[ 480\, \Delta^{35}\, n^{15}\, \left(\Delta n\log(\Delta n) + \log(\varepsilon^{-1})\right),
\]
of which the two-stage direct construction is
responsible for a factor of $480\, \Delta^{12} n^7$.
\label{thm:flip-mixing-time}
\end{thm}

Section~\ref{sec:solution} relates the flip chain
and the connected switch chain, while Section~\ref{sec:disconnected}
relates the connected switch chain and the switch chain.
The proof of Theorem~\ref{thm:flip-mixing-time} can be found at the
end of Section~\ref{sec:disconnected}.
The rest of this section contains some background material, some helpful
graph theory results, precise definitions of the flip and switch Markov
chains, and some definitions regarding canonical paths.

We remark that a ``one-stage direct analysis''
of the flip chain seems difficult.  Furthermore, the multicommodity flow
defined for the switch chain in~\cite{cooper05sampling} does not
apply directly to the connected switch chain, since the flow may pass
through disconnected graphs, even when both end-states are connected.
It may be possible to circumvent this problem by using an alternative
multicommodity flow argument, but we do not explore that option here.

\subsection{Notation, terminology and history}

We will use the notation $\Switch{a}{b}{c}{d}$ to denote the switch operation
as illustrated in Figure~\ref{fig:switchAndFlip}.

The switch chain $\ms$ applies a random switch at each step,
and the flip chain $\mf$ applies a random flip at each step.
(We define both chains formally in Section~\ref{s:chains}.)
The state space $\Omega_S$ for the switch chain is the set
of all $\Delta$-regular graphs on the vertex set $[n]=\{1,2,\ldots, n\}$,
while the state space $\Omega_F$ for the flip chain is the set
of all connected graphs in $\Omega_S$.

The switch chain is known to be ergodic, and has uniform stationary
distribution on $\Omega_S$.
Mahlmann and Schindelhauer~\cite{schindelhauer05kflipper}
showed that the flip chain is ergodic on $\Omega_F$.
They also proved that the transitions of the flip chain are symmetric,
so that the stationary distribution of the flip chain is uniform on
$\Omega_F$.
(We use a slightly different transition procedure, described in
Section~\ref{s:chains}, which is also symmetric.)

The mixing time of a Markov chain is a measure of how many steps
are required before the Markov chain has a distribution
which is close to stationary.
Kannan et al.~\cite{kannan99markovbipartite} considered the mixing time of the switch
Markov chain on ``near regular'' bipartite graphs. Cooper et
al.~\cite{cooper05sampling} gave a polynomial bound on the mixing time $\mix[CDG]$
of the switch chain for all regular simple graphs. The result of~\cite{cooper05sampling}
allows the degree $\Delta=\Delta(n)$ to grow arbitrarily with $n$: it is not restricted
to regular graphs of constant degree.

In Section~\ref{sec:solution} we extend the result of~\cite{cooper05sampling}
from switches to flips, using a method which we
call \emph{two-stage direct canonical path construction}.
We restrict our attention to even degree, 
since this allows the network to have any number of vertices (even or odd).
So for the rest of the paper we assume that $\Delta=2r$ for some $r\geq 2$,
unless otherwise stated.
It happens that even-degree graphs have some other
desirable properties, as will we see in Section~\ref{s:useful}, which
will simplify parts of our argument. However, we see no reason why
the flip chain would fail to be rapidly mixing for odd degrees.

Mahlmann and Schindelhauer~\cite{schindelhauer05kflipper}
 do not provide a bound for the mixing time of
$\mf$, though they do show~\cite[Lemma~9]{schindelhauer05kflipper} that a similar chain with the flip edges
at a distance of
$\bigTh{\Delta^2 n^2 \log{\varepsilon^{-1}}}$ apart will give an $\varepsilon$-approximate expander
in time $\bigO{\Delta n}$. However, these moves are highly non-local and therefore
are unsuitable for the application to self-stabilisation  of decentralised
networks.

An upper bound for the mixing time of $\mf$ was given by Feder et
al.~\cite{saberi06switchflip}
using a comparison argument with the switch Markov chain, $\ms$.
Applying a comparison argument to relate these chains is a difficult
task, since a switch may disconnect a connected graph and a flip
cannot.
Feder et al.~\cite{saberi06switchflip} solve this difficulty
by embedding
$\ms$ into a restricted switch chain which only walks on connected graphs.
They prove that this can be done
in such a way that a path in $\ms$ is only polynomially lengthened.
Their result relaxes the bound in~\cite{cooper05sampling}
by $\bigO{\Delta^{41} n^{45}}$.
By employing a two-stage direct construction,
we give a much tighter result, relaxing the bound on $\mix[CDG]$ by
a factor of $\bigO{r^{12} n^7}$
in the case of $2r$-regular graphs with $r\geq 2$.

Recently, Allen-Zhu et al.~\cite{ABLMO} proved a result related to the work of 
Mahlmann and Schindelhauer~\cite[Lemma 9]{schindelhauer05kflipper}.
They proved that for $\Delta=\Omega(\log n)$, after at most $O(n^2\Delta^2\sqrt{\log n})$ steps of the flip chain, 
the current graph will be an algebraic expander in the sense that the eigenvalue gap of the adjacency matrix 
is $\Omega(\Delta)$.  The proof is based on tracking the change of the graph Laplacian using a 
potential function rather than analysing the flip chain directly,
and does not address the distribution of the output graph, or the 
case $\Delta = o(\log n)$, which 
may be the most convenient in practical applications.

Indeed, it may be possible to reduce the mixing time bound by combining
our analysis with the results of~\cite{ABLMO}.  If $\Delta=o(\log n)$ then 
the bound from Theorem~\ref{thm:flip-mixing-time}, is $O(n^{16}\operatorname{poly}(\log n))$.
Otherwise, when $\Delta = \Omega(\log n)$, the first $O(n^2\Delta^2\sqrt{\log n})$ steps of the flip chain
may be viewed as a preprocessing phase.  As shown by~\cite{ABLMO}, at the end of
this preprocessing phase the current graph is an expander with high probability,
and hence has logarithmic diameter.
The diameter appears in our analysis, 
for example in the proof of Lemma~\ref{lem:lF} below,
and in the analysis of the switch chain~\cite{cooper05sampling}.
Therefore, it is possible that this preprocessing
phase could result in an improvement in the mixing time bound of at least one factor
of the order $\frac{\log n}{n}$.  To achieve this improvement, it appears necessary to
show that a canonical path between two expanders only visits graphs with logarithmic diameter.
We do not pursue this approach here.

\subsection{Graph-theoretical preliminaries}\label{s:useful}

For further details see, for example, Diestel~\cite{diestel}.
Unless otherwise stated, $G$ denotes a simple graph $G = (V,E)$ with
$V=\{1,2,\ldots,n\}$.  The set of neighbours in $G$ of a vertex
$v$ is denoted by $N_G(v)$.

A vertex set $S\subseteq V$ has \emph{edge boundary}
$\partial_E{S} = \{uv\,|\,u \in S, v \notin S\}$, which
contains the edges crossing from inside to outside of $S$.

A graph $G$ is 2-edge-connected if it is connected and every
edge of $G$ lies in a cycle.
More generally, an edge-cut in $G$ is a set $F\subseteq E$ of edges of
$G$ such that
the graph
$G - F$ (which is $G$ with the edges in $F$ deleted)
is disconnected. If $|F|=k$ then $F$ is a $k$-edge-cut.
If vertices $v$, $w$ belong to distinct components of $G-F$ then we
say that $F$ \emph{separates} $v$ and $w$ in $G$.

A \emph{cut vertex} in a graph $G$ is a vertex $v$ such that deleting
$v$ disconnects the graph.

If $p$ is a path from $v$ to $w$ then we say that $p$ is a $(v,w)$-path.
Further, if $p=v\cdots a\cdots b\cdots w$ then we write $p[a:b]$ to denote
the (inclusive) subpath of $p$ with endvertices $a$ and $b$.

We make use of the following structural properties of $2r$-regular graphs
with $r\geq 2$.

\begin{lem}
Let $r\geq 2$ be an integer and let $G=(V,E)$ be a $2r$-regular graph
on $n\geq \max\{ 8,\, 2r+1\}$ vertices.
Then $G$ satisfies the following properties:
\begin{enumerate}
\item[\emph{(i)}]
For any $S \subseteq V$, the edge boundary $\partial_E{S}$ of $S$ is even.
\item[\emph{(ii)}]
Every edge of $G$ lies in a cycle; that is, every connected
component of $G$ is 2-edge-connected.
\item[\emph{(iii)}]
$G$ has fewer than $n/(2r)$ connected components.
\item[\emph{(iv)}]
Let $u,v$ be distinct vertices in the same connected component of $G$.
The number of $2$-edge-cuts in $G$ which separate $u$
and $v$ is at most $n^2/(15 r^2)$.
\end{enumerate}
\label{lem:even-d}
\end{lem}

\begin{proof}
For (i), we argue by contradiction.
Assume that there is an odd
number $k = 2k' + 1$ of edges incident upon some set $S$. Let $G[S]$ be
the subgraph induced by $S$. Let $|S| = s$ and suppose that $G[S]$
has $m$ edges.
Denote the sum of degrees of the vertices in $S$ by $d_G(S)$.
Clearly $d_G(S) = 2rs$. Exactly $k$ edges incident with $S$ have
precisely one endvertex in $S$, so
\[ 2m = d_{G[S]}(S) = 2r s - k = 2rs - 2k' - 1 = 2(rs-k') - 1,\]
but no integer $m$ satisfies this
equation. Hence our assumption is false, proving (i).

For (ii), it suffices to observe that every component
of $G$ is Eulerian, as all vertices have even degree.
Next, (iii) follows since the smallest $2r$-regular graph
is $K_{2r+1}$, and $G$ can have at most $n/(2r+1) < n/(2r)$ components
this small.
The proof of (iv) is deferred to the Appendix, as it is somewhat lengthy.
\end{proof}

\subsection{Canonical paths}\label{s:canonical}

We now introduce some Markov chain terminology and describe
the canonical path method. For further details,
see for example Jerrum~\cite{jerrum03lectures} or
Sinclair~\cite{Sincla92}.

Let $\mathcal{M}$ be a Markov chain with finite state space
$\Omega$ and transition matrix $P$. We suppose that $\mathcal{M}$
is ergodic and has unique stationary distribution $\pi$.
The graph underlying the Markov chain $\mathcal{M}$ is given
by $\mathcal{G}=(\Omega,E(\mathcal{M}))$ where
the edge set of $\mathcal{G}$
corresponds to (non-loop) transitions of $\mathcal{M}$; specifically,
\[ E(\mathcal{M}) = \{ xy  \mid  x,y\in\Omega,\, x\neq y,\, P(x, y) > 0\} \]
We emphasise that  $\mathcal{G}$ has no loops.

Let $\sigma$, $\rho$ be two probablity distributions defined on a finite set $\Omega$.
The \emph{total variation distance} between $\sigma$ and $\rho$,
denoted $\dtv(\sigma,\rho)$, is defined by
\[ \dtv(\sigma,\rho) = \nfrac{1}{2}\sum_{x\in\Omega} |\sigma(x) - \rho(x)|.\]
The \emph{mixing time} $\mix$ of a Markov chain is given by
\[\mix = \max_{x\in{}X} \,\min\,\{T \ge 0\ |\ \dtv(P^t_x,
\pi) \le \varepsilon \mbox{ for all } t \ge T\},\]
where $P^t_x$ is the distribution of the random state $X_t$ of the
Markov chain after $t$ steps with initial state $x$. 
Let the eigenvalues of $P$ be
\[ 1 = \lambda_0 > \lambda_1 \geq \lambda_2 \cdots \geq \lambda_{N-1} \geq -1.\]
Then the mixing time satisfies
\begin{equation}
\label{mixing}
 \tau(\varepsilon) \leq (1-\lambda_\ast)^{-1}\left(\ln(1/\pi^\ast) +
                  \ln(\varepsilon^{-1}) \right)
\end{equation}
where $\lambda_\ast = \max\{ \lambda_1,\, |\lambda_{N-1}|\}$ and
$\pi^\ast = \min_{x\in\Omega}\pi(x)$ is the smallest stationary probability.
(See~\cite[Proposition 1]{Sincla92}.)

The \emph{canonical path method},
introduced by Jerrum and Sinclair~\cite{jerrum89conductancepermanent},
gives a method for bounding $\lambda_1$.
Given any two states $x,y\in\Omega$, a ``canonical'' directed path $\gamma_{xy}$
is defined from $x$ to $y$ in the graph $\mathcal{G}$ underlying the Markov chain,
so that each step in the path corresponds to a transition of the chain.
The \emph{congestion} $\bar{\rho}(\Gamma)$ of the set
$\Gamma = \{ \gamma_{xy} \mid x,y\in\Omega \}$ of
canonical paths is given by
\[ \bar{\rho}(\Gamma) = \max_{e\in E(\M)}  Q(e)^{-1} \sum_{\substack{x,y\in\Omega\\ e\in\gamma_{xy}}}\, \pi(x)\pi(y)|\gamma_{xy}|
\]
where $Q(e) = \pi(u)\, P(u,v)$ when $e=uv\in E(\M)$.
If the congestion along each transition is low then the chain should mix rapidly.
This can be made precise as follows: if $\lambda^\ast = \lambda_1$ then
\begin{equation}
\label{canonicalpath}
(1-\lambda_1)^{-1} \le\ \bar{\rho}(\Gamma).
\end{equation}
(See~\cite[Theorem 5]{Sincla92} or~\cite{jerrum03lectures} for more details.)
In many applications the chain $\mathcal{M}$ is made \emph{lazy} by
replacing its transition matrix $P$ by $(I+P)/2$: this ensures that the
chain has no negative eigenvalues and hence that $\lambda_\ast = \lambda_1$.
However, in many cases it is easy to see that $(1+\lambda_{N-1})^{-1}$
is smaller than the best-known upper bound on $(1-\lambda_1)^{-1}$,
such as that provided by the canonical path method (\ref{canonicalpath}).
In this situation, combining (\ref{mixing}) and (\ref{canonicalpath})
gives an upper bound on the mixing time without the need to make the chain
lazy.

The \emph{multicommodity flow} method is a generalisation of the
canonical path method. Rather than send $\pi(x)\pi(y)$ units
down a single ``canonical'' path from $x$ to $y$, a set $\mathcal{P}_{xy}$
of paths is defined, and the flow from $x$ to $y$ is divided among them.
The congestion can be defined similarly, and low congestion leads to
rapid mixing: see Sinclair~\cite{Sincla92}.

Cooper et al.~\cite{cooper05sampling} gave a multicommodity flow
argument to bound the mixing time of the switch chain $\ms$.
For every pair of distinct states $x,y$, a set of paths from $x$ to $y$ was defined,
one for each vector $(\phi_v)_{v\in [n]}$, where $\phi_v$ is a ``pairing''
of the edges of the symmetric difference
of $x$ and $y$ incident with the vertex $v$.  Then the flow between
$x$ and $y$ was shared equally among all these paths.
In our analysis below, we assume
that a pairing has been \emph{fixed} for each distinct pair of states $x,y$, leading to one
canonical path $\gamma_{xy}$ between them. We will work with the set
of these canonical paths for the switch chain, $\Gamma_S = \{ \gamma_{xy}\}$.
Since our congestion bounds will be independent on the choice of pairing, the
same bound would hold if the original multicommodity flow from~\cite{cooper05sampling}
was used.

\subsection{The switch and flip Markov chains}\label{s:chains}

Given $\Delta\geq 3$ and $n\geq \Delta+1$, the
\emph{switch chain} $\ms$ has state space $\Omega_S$ given by the set of all
$d$-regular graphs on $n$ vertices.
The flip chain $\mf$ has state space $\Omega_F$ given by the set of all
connected $\Delta$-regular graphs on $n$ vertices.
We restrict our attention to regular graphs
of even degree, setting $\Delta=2r$ for some integer $r\geq 2$.

Let $P_S$ (respectively, $P_F$)
denote the transition matrix of the switch (respectively, flip) Markov chain.
The stationary distribution of the switch (respectively, flip)
chain is denoted by $\pi_S$ (respectively, $\pi_F$).
The graph underlying the Markov chain will be denoted by
$\mathcal{G}_S=(\Omega_S,E(\ms))$ and
$\mathcal{G}_F=(\Omega_F,E(\mf))$, respectively.

The transition procedure for the switch Markov chains is
given in Figure~\ref{fig:MarkovSwitchChain}.

\begin{figure}[ht!]
\centering
\fbox{
	\begin{minipage}{10cm}
	\begin{tabbing}
	\hspace{0.6cm}\=\hspace{0.6cm}\=\kill
	From $G \in \Omega_S$ do\\
	\>choose non-adjacent distinct edges $ab$, $cd$, u.a.r.,\\
	\>choose a perfect matching $M$ of $\{a,b,c,d\}$ u.a.r.,\\
	\>if $M \cap E = \emptyset$ then\\
	\>\>delete the edges $ab$, $cd$ and add the edges of $M$,\\
	\>otherwise\\
	\>\>do nothing\\
	\>end if;\\
	end;
	\end{tabbing}
	\end{minipage}
}
\caption{The switch Markov chain $\ms$.}
\label{fig:MarkovSwitchChain}
\end{figure}

The switch chain is irreducible~\cite{petersen,taylor}. Furthermore,
as $P_S(x,x) \geq 1/3$ for all $x\in\Omega_S$, the switch chain
is aperiodic, and hence ergodic. Using an approach of
Diaconis and Saloff-Coste~\cite[p.~702]{DiaSal93}, this also
implies that the smallest eigenvalue $\lambda_{N-1}$ of $\ms$ satisfies
\[ (1+\lambda_{N-1})^{-1}\leq
    \nfrac{1}{2}\, \max_{x\in\Omega_S} P(x,x)^{-1}\leq 3/2.
\]
Since the upper bound on $(1-\lambda_1)^{-1}$
proved in~\cite{cooper05sampling}
is (much) larger than constant, this proves that we can apply the
canonical path method to the switch chain without making the switch
chain lazy.

If $x, y$ are distinct elements of $\Omega_S$ which differ by a switch then
$P_S(x,y) = 1/(3 a_{n,2r})$ where (see~\cite{cooper05sampling})
\begin{equation}
\label{an2r}
 a_{n,2r} = \binom{rn}{2} - n\, \binom{2r}{2},
\end{equation}
Hence the transition matrix $P_S$ is symmetric and the stationary distribution
of the switch chain is uniform over $\Omega_S$.

Figure~\ref{fig:MarkovFlipChain} gives the transition procedure of the
flip chain.
Here we start from a randomly chosen vertex $a$ and perform a random walk
from $a$ of length 3, allowing backtracking.
There are $(2r)^3n$ choices for the resulting walk $(a,b,c,d)$,
and we choose one of these uniformly at random.

\begin{figure}[ht!]
\centering
\fbox{
	\begin{minipage}{10cm}
	\begin{tabbing}
	\hspace{0.6cm}\=\hspace{0.6cm}\=\kill
	From $G \in \Omega_F$ do\\
        \>choose a vertex $a\in [n]$ u.a.r.,\\
	\>choose a walk $(a,b,c,d)$ of length 3 from $a$, u.a.r.,\\
	\>if $|\{ a,b,c,d\}| = 4$ and $E \cap \{ac, bd\} = \emptyset$ then\\
	\>\>delete the edges $ab$, $cd$ and add the edges $ac, bd$,\\
	\>otherwise\\
	\>\>do nothing\\
	\>end if;\\
	end;
	\end{tabbing}
	\end{minipage}
}
\caption{The flip Markov chain $\mf$.}
\label{fig:MarkovFlipChain}
\end{figure}

Mahlmann and Schindelhauer proved that the flip chain is irreducible
on $\Omega_F$~\cite{schindelhauer05kflipper}.
For all states $x\in\Omega_F$ we have $P(x,x) \geq 1/(2r)$,
since in particular the proposed transition is rejected when $c=a$.
Therefore the flip chain is aperiodic, and hence ergodic.
Again, by~\cite[p.~702]{DiaSal93}, this implies
that the smallest eigenvalue $\lambda_{N-1}$ of $\mf$ satisfies
$(1+\lambda_{N-1})^{-1} \leq r$.  This
in turn allows us to apply the canonical path method without making the
chain lazy, since the upper bound we obtain on $(1-\lambda_1)^{-1}$
will be (much) bigger than $r$.

Now suppose that $x, y$ are distinct elements of $\Omega_F$ which differ by a
flip, and that this flip replaces edges $ab, cd$ with $ac, bd$.
Then if exactly one of the edges $bc$, $ad$ is present in $x$, say
$bc$, then the transition from $x$ to $y$ takes place if and only if
the chosen 3-walk is $(a,b,c,d)$ or $(d,c,b,a)$.  Hence in this case
\[ P_F(x,y) = \frac{2}{(2r)^3 n} = P_F(y,x).\]
If both of the edges $bc$ and $ad$ are present then either of them could
take the role of the ``hub edge'' of the flip, and
\[ P_F(x,y) = \frac{4}{(2r)^3 n} = P_F(y,x).\]
Hence $P_F$ is a symmetric matrix and the stationary distribution of the
flip chain is uniform, as claimed.
(Mahlmann and Schindelhauer use a slightly different transition procedure,
and must argue about the number of triangles which contain a given edge
in order to show that their transitions are symmetric:
see~\cite[Lemma 3]{schindelhauer05kflipper}.)

\section{Two-stage direct canonical path construction}
\label{sec:two-stage-direct}

Suppose $\M$ is an ergodic Markov chain whose mixing time we would like to
study, and $\M'$ is a chain for which a canonical path set~$\Gamma'$
has been defined,
with congestion bounded by $\bar{\rho}(\Gamma')$. Let~$P$, $\pi$ and $\Omega$
(respectively, their primed counterparts) be the transition matrix, stationary
distribution and state space of $\M$ (respectively, $\M'$).
Guruswami~\cite{guruswami} proved that under mild conditions on~$\M'$, a bound on the
mixing time of~$\M'$ implies the existence of an appropriate set~$\Gamma'$
with bounded congestion.

We wish to ``simulate'' transition in~$\M'$ using
sequences of transitions in~$\M$, and showing that not too many transitions
in~$\M$ can be involved in any one such simulation. Together with
$\bar{\rho}(\Gamma')$ this results on a bound of the congestion of $\Gamma$, and
hence on the mixing time of $\M$.  Here we assume that $\Omega\subseteq\Omega'$.

If~$\Omega$ is a \emph{strict} subset of~$\Omega'$ then canonical paths
in~$\Gamma'$ that begin and end in~$\Omega$ may travel through states
in~$\Omega'\setminus\Omega$.
(This will be the case in Section~\ref{sec:disconnected}.)
To deal with this, we must define a
surjection~$h\colon \Omega'\longrightarrow\Omega$ 
such that $h(u)=u$ for any $u\in\Omega$, and such that the
preimage of any $y\in\Omega$ is not too large. This means that
each state in~$\Omega$ ``stands in'' for only a few states
in~$\Omega'$.
If $\Omega=\Omega'$ then $h$ is the identity map, by definition.

Suppose that $\gamma'_{wz}\in\Gamma'$ is a canonical path from $w$ to $z$
in $\mathcal{G}'$. Write
\[\gamma'_{wz} = (Z_0, Z_1,\dotsc, Z_j,\, Z_{j+1},\dotsc, Z_\ell),\]
with $Z_0=w,Z_\ell=z$ and $(Z_j,Z_{j+1})\in E(\M')$ for all
$j\in[\ell-1]$.
We now define the \emph{two-stage direct canonical path}~$\gamma_{w,z}$
with respect to the fixed surjection $h:\Omega'\rightarrow \Omega$.

For each $(Z_j,Z_{j+1})\in\gamma'_{wz}$, construct a
path
\[\sigma_{Z_jZ_{j+1}} = (X_0,X_1,\dotsc,X_i, X_{i+1},\dotsc,X_\kappa),\]
with $X_0=h(Z_j)$ and $X_{\kappa}=h(Z_{j+1})$, such that
$X_i\in\Omega$ for $i=0,\ldots, \kappa$
 and $(X_i,X_{i+1})\in E(\M)$ for $i=0,\ldots, \kappa-1$.
We call $\sigma_{Z_j Z_{j+1}}$ a $(\M,\M')$-\emph{simulation path} (or
just a \emph{simulation path}, if no confusion can arise), since this path
simulates a single transition in~$\M'$ using edges of~$\M$.

These simulation paths are concatenated together
to form a canonical path~$\gamma_{wz}$ from $w$ to $z$ in $\mathcal{G}$, as follows:
\[\gamma_{wz} = (\sigma_{Z_0Z_1}, \sigma_{Z_1Z_2}, \dotsc,
                 \sigma_{Z_{\ell-1}Z_\ell}).\]
The following algorithmic interpretation of $\gamma_{wz}$ may be useful as an
illustration. Begin by querying $\gamma'_{wz}$ for the first transition
$(Z_0,Z_1)$ to simulate.  Beginning from $h(Z_0)$, perform transitions in the
corresponding simulation path until state~$h(Z_1)$ is reached. Now query
$\gamma'_{wz}$ for the next transition $(Z_1,Z_2)$ and simulate that; and so on,
until you have simulated all of $\gamma'_{wz}$,  completing the simulation path
$\gamma_{wz}$.

Denote the set of all $(\M,\M')$-simulation paths as
\[\Sigma = \set{\sigma_{xy}\mid (x,y)\in E(\M')}.\]
For these simulation paths to give rise to canonical paths for $\M$
with low congestion, no transition in~$\M$ can feature
in too many simulation paths. For each $t\in E(\M)$, let
\[\Sigma(t) = \set{\sigma\in\Sigma\mid t\in\sigma}\]
be the number of simulation paths containing the transition $t$.
The measures of quality of~$\Sigma$ are the maximum number of simulation paths
using a given transition,
\[B(\Sigma) = \max_{t\in E(\M)} |\Sigma(t)|,\]
and the length of a maximal simulation path,
\[\ell(\Sigma) = \max_{\sigma\in\Sigma} |\sigma|.\]
Note that the simulation paths all depend on the fixed surjection $h$,
and so the same is true of the bounds $B(\Sigma)$ and $\ell(\Sigma)$.
In particular, the quality of the surjection $h$ is measured by
$\max_{y\in\Omega}\, |\{ z\in\Omega' \mid h(z)=y\}$
and this quantity will be needed when calculating $B(\Sigma)$ in particular
(as we will see in Section~\ref{sec:disconnected}).

Some chains are ill-suited to simulation. For instance, a transition in~$\M'$
with large capacity might necessarily be simulated in~$\M$ by a simulation
path that contained an edge with small capacity. Hence we define two
quantities that measure the gap between~$\pi$ and $\pi'$, and $P$ and $P'$.
We define the \emph{simulation gap} $D(\M,\M')$ between $\M$ and $\M'$ by 
\[ D(\M,\M')=\max_{\substack{uv\in E(\M)\\zw\in E(\M')}}\,
    \frac{\pi'(z)\, P'(z,w)}{\pi(u)\, P(u,v)}.\]

We combine these ingredients in the following lemma. The key point
is that the congestion across transitions in $E(\M')$ (and hence across the
simulation paths) is already bounded. If the number of simulation paths making
use of any given transition in $E(\M)$ is bounded, and the graphs are not too
incompatible then the congestion of each two-stage direct canonical path
in~$\Gamma$ will also be bounded.

\begin{thm}
\label{thm:twostagedirect}
Let $\M$ and $\M'$ be two ergodic Markov chains on~$\Omega$ and~$\Omega'$.
Fix a surjection $h:\Omega'\rightarrow\Omega$.
Let~$\Gamma'$ be a set of canonical paths in $\M'$ and let~$\Sigma$ be a set
of $(\M,\M')$-simulation paths defined with respect to $h$.
Let $D=D(\M,\M')$ be the simulation
gap defined above. Then there exists a set $\Gamma$ of canonical paths
in~$\M$ whose congestion satisfies
\[\bar{\rho}(\Gamma)\leq  D\, \ell(\Sigma)\, B(\Sigma)\, \bar{\rho}(\Gamma').\]
\end{thm}

\begin{proof}
We begin with the definition of congestion of the set of paths $\Gamma$ for $\M)$, 
and work towards rewriting all quantities in terms of the set of canonical
paths $\Gamma'$ for $\M'$, as follows:
\begin{align}
\bar{\rho}(\Gamma)&=
	\max_{uv\in E(\M)}\, \frac1{\pi(u)P(u,v)}
	\sum_{\substack{x,y\in\Omega\\uv\in\gamma_{xy}}}
		\pi(x)\pi(y)\,|\gamma_{xy}|\notag\\
	&=\max_{uv\in E(\M)}\, \frac1{\pi(u)P(u,v)}
	\sum_{\substack{zw\in E(\M')\\ \sigma_{zw}\in\Sigma(uv)}}
\:
	\sum_{\substack{x,y\in\Omega\\ \sigma_{zw}\in\gamma_{xy}}}
		\pi(x)\pi(y)\,|\gamma_{xy}|\label{line:path-split}\\
	&\leq \max_{uv\in E(\M)}\, \frac1{\pi(u)P(u,v)}
	\sum_{\substack{zw\in E(\M')\\ \sigma_{zw}\in\Sigma(uv)}}
\:
	\sum_{\substack{x,y\in\Omega\\zw\in\gamma'_{xy}}}
		\pi(x)\pi(y)\,|\gamma_{xy}|.
\label{line:sigmas-switch}
\end{align}
Line~(\ref{line:path-split}) splits the two-stage canonical paths into simulation paths.
By a slight abuse of notation, we write $\sigma_{zw}\in \gamma_{xy}$ to mean that $\sigma_{zw}$ is 
one of the simulation paths which was concatenated together to make $\gamma_{xy}$.
Here we assume that
each edge~$uv$ occurs at most once in any given simulation path $\sigma\in\Sigma$.
Otherwise, the
second use of~$uv$ in any~$\sigma$ forms a cycle which may be pruned: continue
pruning until no edge appears more than once in any simulation path.
We also assume that no self-loops appear in any simulation path
(these may also be pruned.)
Line~(\ref{line:sigmas-switch}) uses the
one-to-one correspondence between $\gamma_{xy}$ and $\gamma'_{xy}$
for all~$x,y\in\Omega$. Note that this correspondence is still one-to-one even
when~$\Omega\subset\Omega'$, since the endpoints are identical (as $h$ acts as
the identity map on $\Omega$) and each
canonical path is determined by its endpoints.

Now we observe that for all $uv\in E(\M)$ and $zw\in E(\M')$,
by definition of the simulation gap $D$,
\[ \frac{1}{\pi(u)P(u,v)} \leq  D\, \frac{1}{\pi'(z)P'(z,w)}.\]
Similarly, for all $x,y\in\Omega'$,
\[ |\gamma_{xy}| \leq \ell(\Sigma)\, |\gamma'_{xy}|.
\]
Therefore, as each summand of (\ref{line:sigmas-switch}) is nonnegative
and $\Omega\subseteq \Omega'$, we obtain
\begin{align}
\bar{\rho}(\Gamma)&\leq D\, \ell(\Sigma)\, \max_{uv\in E(\M)} \,
	\sum_{\substack{zw\in E(\M')\\\sigma_{zw}\in \Sigma(uv)}}
		\frac1{\pi'(z)P'(z,w)}
	\sum_{\substack{x,y\in\Omega'\\zw\in\gamma'_{xy}}}
		\pi'(x)\pi'(y)\,|\gamma'_{xy}|\notag\\
	&\leq D\, \ell(\Sigma)\, \max_{uv\in E(\M)} \,
	\sum_{\substack{zw\in E(\M)\\\sigma_{zw}\in \Sigma(uv)}}
		\bar{\rho}(\Gamma')\notag\\
	&\leq D\, \ell(\Sigma)\, B(\Sigma)\, \bar{\rho}(\Gamma'),\notag
\end{align}
as required.
\end{proof}

We remark that in the above theorem, the surjection $h$ is only needed to construct the 
set of simulation paths (which are given as input to the theorem).

\section{Analysis of the flip chain}\label{sec:solution}

Our aim is to define canonical paths in the
flip chain that simulate canonical paths defined in the switch chain.
There are two main differences between the switch and flip chains which
cause difficulties.
The state space of the flip chain is restricted to connected graphs,
and no flip can disconnect a connected graph.  However, a
switch can increase or decrease the number of components in a graph by one.
Hence, from a connected graph, the set of available flips may be
strictly smaller than the set of available switches.
Following~\cite{saberi06switchflip} we
overcome these difficulties in two steps, using an intermediate chain to
bridge the gap.

Define the \emph{connected switch chain} $\msc$ to be the
projection of the switch chain~$\ms$ onto state space~$\OF$.
A transition of $\msc$ is a \emph{connected switch}, which is a switch on a
connected graph that preserves connectedness.
To be precise, if
$\PSC$ denotes the transition matrix of $\msc$ and $\PS$ denotes the
transition matrix of $\ms$, then for all $x\neq y\in\OS$,
\begin{equation}
\label{conn-switch}
\PSC(x,y) = \begin{cases} \PS(x,y)&\text{if }x,y\in\OF,\\
                               0&\text{otherwise}.
\end{cases}
\end{equation}
Any transition of $\ms$ which produces a disconnected graph is replaced by
a self-loop in $\msc$. Hence,
for all~$x\in\OF$,
\[\PSC(x,x) = \PS(x,x)+\sum_{y\in\OS\setminus\OF} \PS(x,y).\]
Transitions are symmetric, and so the stationary distribution $\piSC$ of
$\msc$ is the uniform stationary distribution on $\Omega_F$ (that is,
$\piSC = \piF$).

We proceed as follows. In Section~\ref{sec:long-flip}
we show that a connected switch can be
simulated by sequences of flips, which we will call a \emph{long-flip}.
This gives a set of $(\mf,\msc)$-simulation paths. Here the surjection
$h$ is the identity map, since these chains share the same state space.
We apply Theorem~\ref{thm:twostagedirect} to this set of simulation paths.

Then in Section~\ref{sec:disconnected} we apply the theorem again to
relate $\msc$ and $\ms$. We use a construction of Feder et
al.~\cite{saberi06switchflip}
to define a surjection $h:\Omega_S\rightarrow \Omega_F$. We then define
$(\msc,\ms)$-simulation paths with respect to this surjection
and apply Theorem~\ref{thm:twostagedirect}.
Theorem~\ref{thm:flip-mixing-time} is proved
by combining the results of these two sections.

The two-stage direct method has lower cost, in the sense of relaxation of the
mixing time, than the comparison method used in \cite{saberi06switchflip}.
Both approaches have at their heart the construction of paths in $\mf$ to
simulate transitions in~$\ms$, but the two-stage canonical path analysis
is more direct and yields significantly lower mixing time bounds
than those obtained in~\cite{saberi06switchflip}.

\subsection{Simulating a connected switch: the long-flip} \label{sec:long-flip}

We must define a set of $(\mf,\msc)$-simulation paths,
as defined in Section~\ref{sec:two-stage-direct}.
The two Markov chains have the same state space $\Omega_F$, and so
we use the identity map as the surjection.

Let $S = (Z,Z')$ be a connected switch which we wish to simulate using
flips.
Then $Z$ and $Z'$ are connected graphs which
differ by exactly four edges, two in $Z\setminus Z'$ and two
in~$Z'\setminus Z$. We say that the two edges $ab, cd\in Z\setminus Z'$
have been \emph{switched out} by $S$, and the two edges
$ac, bd\in Z'\setminus Z$ have been \emph{switched in} by $S$.
Given an arbitrary labelling of one switch edge as $ab$ and the other
as $cd$, these four permutations of the vertex labels
\begin{equation}
 (\, ),\quad (a\, b)(c\, d),\quad (a\, c)(b\, d),\quad (a\, d)(b\, c)
\label{perms}
\end{equation}
preserve the edges switched in and out by $\Switch{a}{b}{c}{d}$.

To simulate the connected switch~$S=(Z,Z')$ we will define a
simulation path called a \emph{long-flip}, which consists of
a sequence of flips
\[\flippath_{Z Z'} = (X_0, X_1, \dotsc, X_{i-1},X_{i} \dotsc, X_\kappa)\]
with $X_0=Z$ and $X_\kappa = Z'$, such that $(X_{i-1},X_i)\in E(\mf)$
for $i=1,\ldots, \kappa$. We define these simulation paths in this
section, and analyse their congestion in Section~\ref{sec:longflip-congestion}.

A long-flip uses a \emph{hub path} $p=(b,p_1,p_2,\ldots,p_{\nu-1},c)$ from $b$
to $c$.
In the simplest case (in which $p$ does not contain~$a$
or $d$ and no internal vertex of~$p$ is adjacent to either of these vertices) we will
simply ``reverse'' $p$ using flips, 
as illustrated later in Figure~\ref{fig:switchAsFlips}.
Here by ``reversing $p$'' we imagine removing the path $p$ with vertex $b$ at the ``left''
and $c$ at the ``right''.  Now remove $p$ from the graph and
replacing it in the reverse orientation, but with $a$ reattached to the left-most
vertex of $p$, which is now $c$, and with $d$ reattached to the right-most vertex of $p$,
which is now $b$.
This has the effect of deleting the
edges~$ab,cd$ and replacing them with~$ac,bd$, as desired.
This process is described in more detail in Section~\ref{case1} below 
The choice of hub path is important, and we now explain how the hub
path will be (deterministically) selected.

Suppose that $\Switch{a}{b}{c}{d}$ is the connected switch which takes $Z$ to $Z'$,
which we wish to simulate using flips.
We say that a path $p$ is a \emph{valid hub path} for the switch~$\Switch{a}{b}{c}{d}$
in $Z$
if it satisfies the following \emph{$(a,b,c,d)$~path conditions}:
\begin{enumerate}
\setlength\itemsep{-0.5mm}
  \item[(i)] $p$ is a $(b,c)$-path in $Z$;
  \item[(ii)] $a$ and $d$ do not lie on $p$; and
  \item[(iii)] the vertices of $p$ are only adjacent through path edges; that is, there
are no edges in $Z$ of the form~$p_ip_j$ with $|i-j|\neq 1$
for~$i,j\in\set{0,\dotsc,|p|}$.
\end{enumerate}
We refer to condition~(iii) as the \emph{no-shortcut property}.
Given a path $p$ which satisfies conditions (i) and (ii), it is easy to find
a path with the no-shortcut property, by simply following the shortcuts. 
A shortest $(b,c)$-path is a valid $(a,b,c,d)$~path provided it avoids $a$
and $d$.

Given $Z$ and the switch edges $ab,cd$, we perform the following procedure:
\begin{itemize}
\item
Let $p$ be the lexicographically-least shortest $(b,c)$-path in $Z$ which
avoids $a$ and $d$, if one exists;
\item If not, let $p$ be the lexicographically-least
shortest $(a,d)$-path in $Z$ which avoids $b$ and $c$, if one exists, and apply
the vertex relabelling $(a\, b)(c\, d)$.  (The result of this relabelling is that
now $p$ is the lexicographically-least shortest $(b,c)$-path in $Z$ which avoids
$a$ and $d$.)
\item
If a path $p$ was found, return the path $p$. Otherwise, return FAIL.
\end{itemize}

If a path $p$ was found by the above procedure then we will say that we are in Case~1.
Otherwise the output is ``FAIL'', and we say that we are in Case~2.

Suppose that if $p$ is a lexicographically-least shortest path between
vertices $\alpha$ and $\beta$ in a graph $Z$.  We treat $p$ as a
directed path from $\alpha$ to $\beta$, with $\alpha$ on the ``left''
and $\beta$ on the ``right''.  Suppose that
vertices $\gamma$, $\delta$ lie on $p$, with $\gamma$ to the ``left'' of
$\delta$.  Then the lexicographically-least shortest $(\gamma,\delta)$-path
in $Z$ is a subpath of $p$ (or else we contradict the choice of $p$).
This important property, satisfied by hub paths, will be called the
\emph{recognisability property}. It will enable us to reconstruct the hub
path cheaply, as we will see in Section~\ref{sec:longflip-congestion}.

We now define the simulation path for Case 1 and Case 2 separately.
Recall that since flips cannot disconnect a connected graph, the simulation
paths always remain within $\Omega_F$.

\subsubsection{Case 1: the long-flip, and triangle-breaking}\label{case1}

We must define a simulation path $\sigma_{ZZ'}$ for the connected
switch from $Z$ to $Z'$ given by $\Switch{a}{b}{c}{d}$.
The path $p$ selected above (the lexicographically-least shortest
$(b,c)$-path which avoids $a$ and $d$) is a valid hub path
for this switch, and we refer to it as the hub path.
If $p$ has length 1 then this switch is a flip, and so we assume that
$p$ has length at least two.

We simulate the connected switch $\Switch{a}{b}{c}{d}$ by reversing the path $p$,
as shown in Figure~\ref{fig:switchAsFlips}.
In (A), the switch $\Switch{a}{b}{c}{d}$ is shown with alternating solid
and dashed edges, with solid edges belong to $Z$ and dashed edges
belonging to $Z'$.
Write the hub path as $p = (b, p_1, \dots, p_{\nu-1},c)$, where $b=p_0$ and
$c=p_{\nu}$.

\begin{figure}[ht!]
\begin{center}
\begin{tikzpicture}
\draw [fill] (0,1) circle (0.1);
\node [left] at (-0.2,1.0) {$a$};
\draw [-] (0,0.1) -- (0.1,0) -- (0,-0.1) -- (-0.1,0) -- (0,0.1); 
\node [left] at (-0.2,0) {$b$};
\draw [fill] (5,0) circle (0.1);
\node [right] at (5.2,0) {$c$};
\draw [fill] (5,1) circle (0.1);
\node [right] at (5.2,1) {$d$};
\node [above] at (2.5,1.0) {$(A)$};
\draw [-] (0,0.1) -- (0,1);
\draw [-] (5,0) -- (5,1);
\draw [dashed,-] (0,1) -- (5,0.071);
\draw [dashed,-] (5,1) -- (0.071,0.071);
\draw [-] (0.1,0) -- (2,0);
\draw [-] (4,0) -- (5,0);
\draw [dotted,-] (2,0) -- (4,0);
\draw [fill] (1,0) circle (0.1);
\node [below] at (1.0,-0.2) {$p_1$};
\draw [fill] (2,0) circle (0.1);
\node [below] at (2.0,-0.2) {$p_2$};
\draw [fill] (4,0) circle (0.1);
\node [below] at (4.0,-0.2) {$p_{\nu-1}$};
\node at (6.5,0.5) {{\Large $\Rightarrow_F$}};
\draw [fill] (8,1) circle (0.1);
\node [left] at (7.8,1.0) {$a$};
\draw [fill] (8,0) circle (0.1);
\node [left] at (7.8,0) {$p_1$};
\draw [fill] (13,0) circle (0.1);
\node [right] at (13.2,0) {$c$};
\draw [fill] (13,1) circle (0.1);
\node [right] at (13.2,1) {$d$};
\node [above] at (10.5,1.0) {$(B)$};
\draw [-] (8,0) -- (8,1);
\draw [-] (13,0) -- (13,1);
\draw [-] (8,0) -- (8.9,0); 
\draw [-] (9.1,0) -- (10,0); 
\draw [-] (12,0) -- (13,0);
\draw [dotted,-] (10,0) -- (12,0);
\draw [-] (9,0.1) -- (9.1,0) -- (9,-0.1) -- (8.9,0) -- (9,0.1); 
\node [below] at (9.0,-0.2) {$b$};
\draw [fill] (10,0) circle (0.1);
\node [below] at (10.0,-0.2) {$p_2$};
\draw [fill] (12,0) circle (0.1);
\node [below] at (12.0,-0.2) {$p_{\nu-1}$};
\node at (0,-2.5) {{\Large $\Rightarrow_F^\ast$}};
\draw [fill] (1.5,-2) circle (0.1);
\node [left] at (1.3,-2.0) {$a$};
\draw [-] (1.5,-2.9) -- (1.6,-3) -- (1.5,-3.1) -- (1.4,-3) -- (1.5,-2.9); 
\node [left] at (1.3,-3) {$p_1$};
\draw [fill] (6.5,-3.0) circle (0.1);
\node [right] at (6.7,-3.0) {$b$};
\draw [fill] (6.5,-2) circle (0.1);
\node [right] at (6.7,-2) {$d$};
\node [above] at (4.0,-2.0) {$(C)$};
\draw [-] (1.5,-2.9) -- (1.5,-2);
\draw [-] (6.5,-3) -- (6.5,-2);
\draw [-] (1.6,-3) -- (3.5,-3);
\draw [-] (5.5,-3) -- (6.5,-3);
\draw [dotted,-] (3.5,-3) -- (5.5,-3);
\draw [fill] (2.5,-3) circle (0.1);
\node [below] at (2.5,-3.2) {$p_2$};
\draw [fill] (3.5,-3) circle (0.1);
\node [below] at (3.5,-3.2) {$p_3$};
\draw [fill] (5.5,-3) circle (0.1);
\node [below] at (5.5,-3.2) {$c$};
\node at (8.0,-2.5) {{\Large $\Rightarrow_F^\ast$}};
\draw [fill] (9.5,-2) circle (0.1);
\node [left] at (9.3,-2.0) {$a$};
\draw [fill] (9.5,-3) circle (0.1);
\node [left] at (9.3,-3) {$c$};
\draw [fill] (14.5,-3.0) circle (0.1);
\node [right] at (14.7,-3.0) {$b$};
\draw [fill] (14.5,-2) circle (0.1);
\node [right] at (14.7,-2) {$d$};
\node [above] at (12.0,-2.0) {$(D)$};
\draw [-] (9.5,-3) -- (9.5,-2);
\draw [-] (14.5,-3) -- (14.5,-2);
\draw [-] (9.5,-3) -- (11.5,-3);
\draw [-] (13.5,-3) -- (14.5,-3);
\draw [dotted,-] (11.5,-3) -- (13.5,-3);
\draw [fill] (10.5,-3) circle (0.1);
\node [below] at (10.5,-3.2) {$p_{\nu-1}$};
\draw [fill] (11.5,-3) circle (0.1);
\node [below] at (11.5,-3.2) {$p_{\nu-2}$};
\draw [fill] (13.5,-3) circle (0.1);
\node [below] at (13.5,-3.2) {$p_1$};
\end{tikzpicture}
\caption{A hub path reversed by flips.}
\label{fig:switchAsFlips}
\end{center}
\end{figure}
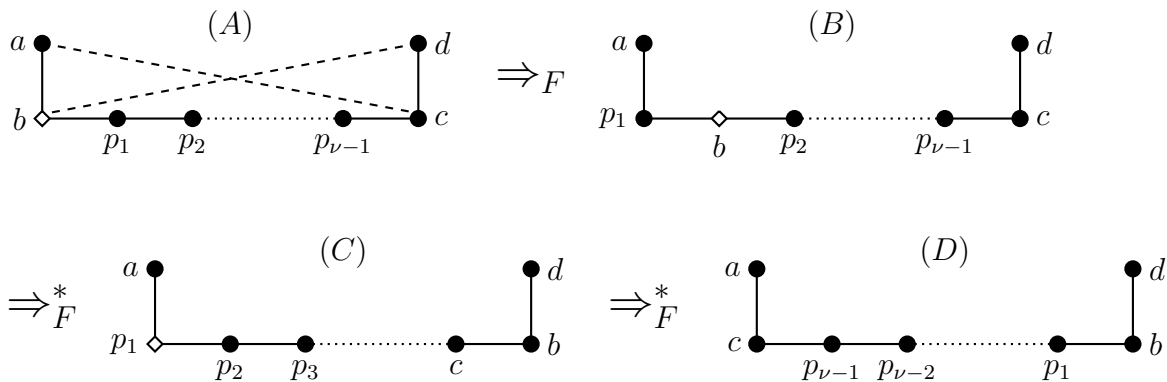

The reversal of~$p$ proceeds in \emph{runs}, in which successive vertices are
moved along the path one at a time into their final position. The run that
moves vertex~$p_i$ into place is called the \emph{$p_i$-run}, during which
$p_i$ is referred to as the \emph{bubble} due to the way it percolates along
path~$p$.
The first stage of the path reversal is the $b$-run, which moves~$b$ adjacent to
$d$. Next is the~$p_1$-run, and then the~$p_2$-run, until the
$p_{\nu-2}$-run, which puts~$c$ adjacent to~$a$ and completes the reversal
of~$p$.
Specifically, the $p_i$-run consists of the sequence of flips
\[(a,p_{i},p_{i+1},p_{i+2}), (p_{i+1},p_{i},p_{i+2},p_{i+3}),
(p_{i+2},p_{i},p_{i+3},p_{i+4}), \dotsc, (p_{\nu-1},p_{i},c,p_{i-1}),\]
where $p_0=b$, $p_\nu=c$ (and in the case of the $b$-run, $p_{-1}$ is
understood to denote~$d$).
After the~$p_i$ run, $c$ has moved one step closer to $a$ on path~$p$.
See Figure~\ref{fig:switchAsFlips}, where the current bubble is shown as a diamond.
The result of the first flip of the $b$-run is shown
in (B), and the situation at the end of the $b$-run is shown in (C),
At this stage, $b$ is in its final
position, adjacent to $d$.  Finally in (D) we see the graph after all runs
are complete, with the path reversed.
(Note that $\Rightarrow_F$ denotes a single flip while $\Rightarrow_F^\ast$
denotes a sequence of zero or more flips.)

Since the distance from $c$ to $a$ decreases by one after each run,
the total number of flips used to reverse the path $p$ is
$\frac12\nu(\nu+1)$. The no-shortcut property of hub paths is
required so that no flip in any run is blocked by parallel edges within~$p$.

However, the first flip of the $b$-run is blocked if $p_1$ is adjacent to $a$.
Before we can reverse the path $p$, we must break this \emph{triangle}~$[a,b,p_1]$
using the following procedure. Choose the least-labelled vertex~$x$ that is adjacent
to $b$ but not~$p_1$. The $(a,b,c,d)$ path
conditions imply that $p_2$ is adjacent to $p_1$, but not to $b$, so the
existence of $x$ follows from $2r$-regularity. Then the triangle~$[a,b,p_1]$ can
be removed by performing a flip on the path~$(x,b,p_1,p_2)$, as shown below.

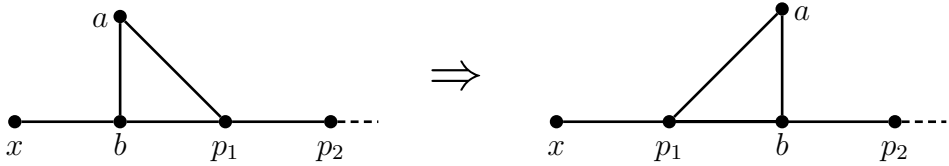
\begin{center}
\begin{tikzpicture}[scale=1.4,inner sep=0.5mm,line width=1pt]
\pgfsetxvec{\pgfpoint{1cm}{0cm}}
\pgfsetyvec{\pgfpoint{0cm}{1cm}};
\path
(0,0) node [circle,draw,fill=black,label=below:\st $x$] (x){}
(1,1) node [circle,draw,fill=black,label=left:\st $a$] (a){}
(1,0) node [circle,draw,fill=black,label=below:\st $b$] (b){}
(2,0) node [circle,draw,fill=black,label=below:\st $p_1$] (p_1){}
(3,0) node [circle,draw,fill=black,label=below:\st $p_2$] (p_2){}
(3.5,0) node (c){};
\draw (x)--(b)--(p_1)--(p_2) (b)--(a)--(p_1);
\draw[densely dashed] (p_2)--(c);
\end{tikzpicture}
\raisebox{1cm}{\quad\LARGE $\Rightarrow$\quad}
\begin{tikzpicture}[scale=1.5,inner sep=0.5mm,line width=1pt,]
\pgfsetxvec{\pgfpoint{1cm}{0cm}}
\pgfsetyvec{\pgfpoint{0cm}{1cm}};
\path
(0,0) node [circle,draw,fill=black,label=below:\st $x$] (x){}
(2,1) node [circle,draw,fill=black,label=right:\st $a$] (a){}
(1,0) node [circle,draw,fill=black,label=below:\st $p_1$] (p_1){}
(2,0) node [circle,draw,fill=black,label=below:\st $b$] (b){}
(3,0) node [circle,draw,fill=black,label=below:\st $p_2$] (p_2){}
(3.5,0) node (c){};
\draw (x)--(b)--(p_1)--(p_2) (b)--(a)--(p_1);
\draw[densely dashed] (p_2)--(c);
\end{tikzpicture}
\end{center}

Now $p_1$ is still adjacent to $a$, but it is no longer on the $(b,c)$-path,
and the length of the $(b,c)$-path has been reduced by $1$. However, it is
possible that $p_2$ is also adjacent to $a$, so we may still have a triangle.
If so, we must repeat the process, as indicated below, successively shifting
$b$ to the right on the $(b,c)$-path, and reducing its length.
To break the triangle $[a,b,p_2]$ we flip on $(p_1,b,p_2,p_3)$.

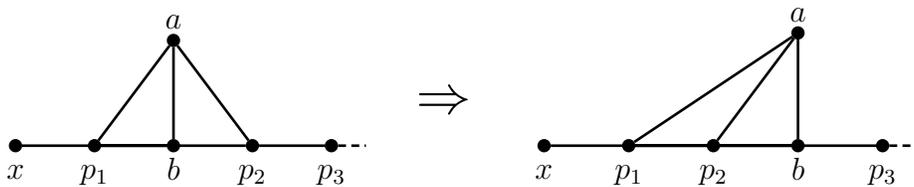
\begin{center}
\begin{tikzpicture}[scale=1.4,xscale=0.75,inner sep=0.5mm,line width=1pt,]
\pgfsetxvec{\pgfpoint{1cm}{0cm}}
\pgfsetyvec{\pgfpoint{0cm}{1cm}};
\path
(0,0) node [circle,draw,fill=black,label=below:\st $x$] (x){}
(2,1) node [circle,draw,fill=black,label=above:\st $a$] (a){}
(1,0) node [circle,draw,fill=black,label=below:\st $p_1$] (p_1){}
(2,0) node [circle,draw,fill=black,label=below:\st $b$] (b){}
(3,0) node [circle,draw,fill=black,label=below:\st $p_2$] (p_2){}
(4,0) node [circle,draw,fill=black,label=below:\st $p_3$] (p_3){}
(4.5,0) node (c){};
\draw (x)--(b)--(p_1)--(p_2)--(p_3) (b)--(a)--(p_1) (a)--(p_2);
\draw[densely dashed] (p_3)--(c);
\end{tikzpicture}
\raisebox{1cm}{\quad\LARGE $\Rightarrow$\quad}
\begin{tikzpicture}[scale=1.5,xscale=0.75,inner sep=0.5mm,line width=1pt,]
\pgfsetxvec{\pgfpoint{1cm}{0cm}}
\pgfsetyvec{\pgfpoint{0cm}{1cm}};
\path
(0,0) node [circle,draw,fill=black,label=below:\st $x$] (x){}
(3,1) node [circle,draw,fill=black,label=above:\st $a$] (a){}
(1,0) node [circle,draw,fill=black,label=below:\st $p_1$] (p_1){}
(2,0) node [circle,draw,fill=black,label=below:\st $p_2$] (p_2){}
(3,0) node [circle,draw,fill=black,label=below:\st $b$] (b){}
(4,0) node [circle,draw,fill=black,label=below:\st $p_3$] (p_3){}
(4.5,0) node (c){};
\draw (x)--(b)--(p_1)--(p_2)--(p_3) (b)--(a)--(p_1) (a)--(p_2);
\draw[densely dashed] (p_3)--(c);
\end{tikzpicture}
\end{center}

More generally, if the edges $a p_1,\ldots, a p_k$ are all present
but the edge $a p_{k+1}$ is not
then the $k$'th triangle $[a,b,p_k]$ is broken with the flip ($p_{k-1},b,p_k,p_{k+1}$).
After this flip, $b$ is adjacent to $p_k$ and $p_{k+1}$.
We call these \emph{triangle-breaking flips}.

The no-shortcut property of the $(b,c)$-path guarantees that all the
triangle-breaking flips after the first one are valid. The triangle-breaking
process must end when $bc$ becomes an edge, at the very latest.
This implies that $k\leq \nu-1$.
Furthermore, since $a$ has degree $2r$ we also have $k\leq 2r-1$.
When the process terminates, the edge $ap_{k+1}$ is absent and $bp_{k+1}$
is the first edge on the $(b,c)$-path. Now the long-flip can be performed
on this new $(b,c)$ hub path.  If $k=\nu-1$ then $p_{k+1}=c$ and this
long-flip is a single flip.

After the path is fully reversed we must undo the triangle removals,
by reversing each triangle-breaking flip, in reverse order.
That is, if $k$ triangle-breaking flips were performed, then performing
these flips in order will reinstate these triangles:
\[ (p_{k-1},p_k,b,p_{k+1}),\,\, (p_{k-2}, p_{k-1}, b, p_k),\ldots
        (p_1,p_2,b,p_3),\,\,\, (x, p_1,b,p_2).
\]
(When $k=1$, only the last flip is performed.)
We call these \emph{triangle-restoring flips}.
These flips are not blocked
because the only edges whose presence could block them are guaranteed not to
exist by the no-shortcut property of~$p$, and by the fact that these edges
were removed by the triangle-breaking flips.

\subsubsection{Case 2: The detour flip}\label{case2}

In Case~2, we know that there are no paths from $b$ to $c$ in $Z$ which avoid
both $a$ and $d$,  and that there are no paths from $a$ to $d$ in
$Z$ which avoid $b$ and $c$.
It follows that $Z$ must have the structure illustrated in Figure~\ref{fig:Case2structure}.
\begin{figure}[ht!]
\begin{center}
\begin{tikzpicture}
\node [above] at (0.3,3.7) {$Z_{ab}$};
\node [above] at (3.7,3.7) {$Z_{ac}$};
\node [below] at (3.7,0.3) {$Z_{cd}$};
\node [below] at (0.3,0.3) {$Z_{bd}$};
\begin{scope}[shift={(3.05,3.05)}]  
  \begin{scope}[rotate=-45]
    \begin{scope}[shift={(-3,-3)}]
   \draw[gray,dashed,fill=lightgray,fill opacity=0.5] (3,3) ellipse (1.30cm and 0.6cm);
     \end{scope}
   \end{scope}
\end{scope}
\begin{scope}[shift={(3.05,0.95)}]  
  \begin{scope}[rotate=45]
    \begin{scope}[shift={(-3,-1)}]
   \draw[gray,dashed,fill=lightgray,fill opacity=0.5] (3,1) ellipse (1.30cm and 0.6cm);
     \end{scope}
   \end{scope}
\end{scope}
\begin{scope}[shift={(0.95,0.95)}]  
  \begin{scope}[rotate=-45]
    \begin{scope}[shift={(-1,-1)}]
   \draw[gray,dashed,fill=lightgray,fill opacity=0.5] (1,1) ellipse (1.30cm and 0.6cm);
     \end{scope}
   \end{scope}
\end{scope}
\begin{scope}[shift={(0.95,3.05)}]  
  \begin{scope}[rotate=45]
    \begin{scope}[shift={(-1,-3)}]
   \draw[gray,dashed,fill=lightgray,fill opacity=0.5] (1,3) ellipse (1.30cm and 0.6cm);
     \end{scope}
   \end{scope}
\end{scope}
\draw [thick,-] (0.0,2) -- (2.0,4.0);
\draw [thick,-] (2.0,0) -- (4.0,2.0);
\draw [fill] (2.0,0.0) circle (0.13);
\node [below] at (2.0,-0.2) {$d$};
\draw [fill] (0.0,2) circle (0.13);
\node [left] at (-0.2,2) {$b$};
\draw [fill] (2,4) circle (0.13);
\node [above] at (2,4.2) {$a$};
\draw [fill] (4.0,2.0) circle (0.13);
\node [right] at (4.2,2.0) {$c$};
\end{tikzpicture}
\end{center}
\caption{The structure of a graph $Z$ in Case 2.}
\label{fig:Case2structure}
\end{figure}
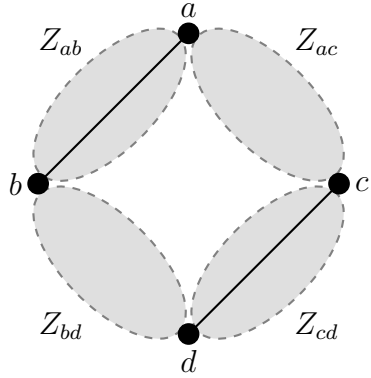
In this figure, each $Z_{ij}$ (shown as a grey ellipse) is an induced subgraph of $Z$,
with $i,j\in Z_{ij}$.
There are no edges from $Z_{ij}\setminus \{ i,j\}$ to $Z_{k\ell}\setminus \{ k,\ell\}$
whenever $\{i,j\}\neq \{k,\ell\}$, by the assumption that $Z$ is in Case 2.  Furthermore,
$Z_{ab}$ and $Z_{cd}$ are both connected, since $Z$ is connected and $ab, cd\in E(Z)$.

At most one of $Z_{ac}$ or $Z_{bd}$ may be empty or
disconnected, since $Z$ is connected.
A deterministic procedure will now be described
for constructing a 3-path $(u,d,w,z)$ in $Z$ which will be used
for the first flip in the simulation path.

There are two subcases to consider for the choice of $w$:
\begin{itemize}
\item
First suppose that, after applying the relabelling $(a\, b)(c\, d)$ if necessary,
$Z_{ac}$ is empty or disconnected.
Since $cd$ is not a bridge in $Z$, by Lemma~\ref{lem:even-d}(ii), there
exists a cycle in $Z_{cd}$ which passes through the edge $cd$. It follows that
$d$ has at least one neighbour in $Z_{cd}\setminus c$ which is connected to $c$ by a
path in $Z_{cd}\setminus cd$. Let $w$ be the least-labelled such neighbour of $d$
in $Z_{cd}\setminus cd$.
\item
Secondly, suppose that both $Z_{ac}$ and $Z_{bd}$ are nonempty and connected.
(So far, we have not performed any relabelling.)
Then at most one of $Z_{ab} \setminus ab$, $Z_{cd}\setminus cd$ may be
disconnected, or else the graph $Z'$ will be disconnected,
a contradiction. Applying the relabelling $(a\, c)(b\, d)$ if necessary,
we can assume that $Z_{cd} \setminus cd$ is connected.
Hence there is a path from $c$ to $d$ in $Z_{cd}\setminus cd$,
so $d$ has a neighbour in $Z_{cd}\setminus c$ which is joined to $c$ by a
path in $Z_{cd}\setminus cd$.  Let $w$ be the least-labelled such neighbour
of $d$ in $Z_{cd}\setminus c$.
\end{itemize}
In both subcases, we have established the existence of a suitable vertex $w$.
Next, observe that since (in both the above subcases)
$Z_{bd}$ is nonempty and connected,  there is a path from
$b$ to $d$
in $Z_{bd}$ of length at least two (as $bd\not\in E(Z)$). This implies
that $d$ has a neighbour in $Z_{bd}$ which is joined to $b$ by a path in
$Z_{bd}\setminus d$.
Let $u$ be the least-labelled such neighbour of $d$ in $Z_{bd}\setminus b$.
Finally, since all neighbours of $w$ belong to $Z_{cd}$ and $u\not\in Z_{cd}$,
it follows that $\{ u,w\}\subseteq N_Z(d)\setminus N_Z(w)$.
Since $Z$ is regular, this implies that
\[ |N_Z(w)\setminus N_Z(d)| = |N_Z(d)\setminus N_Z(w)| \geq 2.\]
Hence there is some vertex other than $d$ in
$N_Z(w)\setminus N_Z(d)$.  Let $z$ be the least-labelled such vertex.
Then $z\in Z_{cd}\setminus \{c,d\}$, since $c\in N_Z(d)$ and every neighbour of
$w$ belongs to $Z_{cd}$.

Now we can define the simulation path $\sigma_{ZZ'}$ from $X_0=Z$ to $X_\kappa = Z'$.
The first step in the simulation path is the flip on the path $(u,d,w,z)$,
giving the graph $X_1$.
We call this the \emph{detour flip}, as it will give us a way to
avoid $d$.  See the first line of Figure~\ref{fig:Case2}.

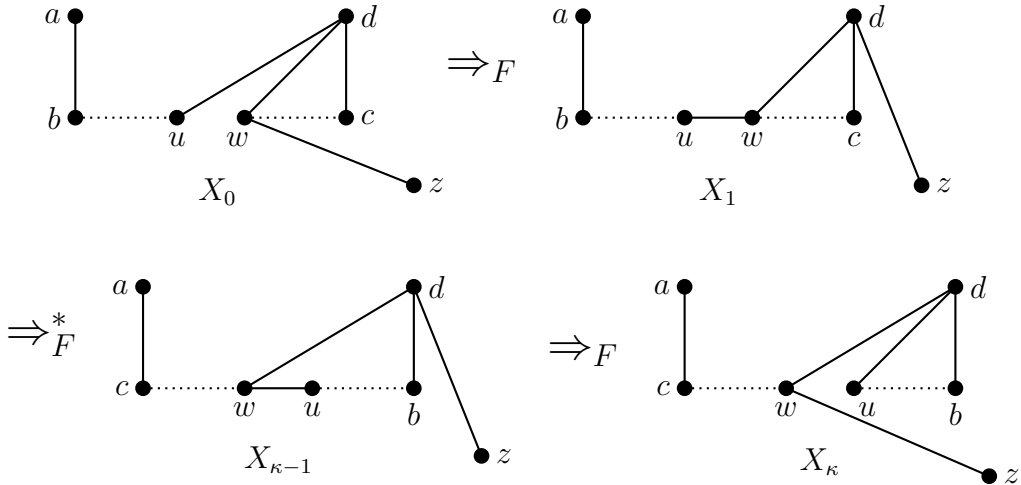
\begin{figure}[ht!]
\begin{center}
\begin{tikzpicture}[scale=0.9]
\draw [fill] (0.5,2.5) circle (0.1);
\node [left] at (0.3,2.5) {$a$};
\draw [fill] (0.5,1) circle (0.1);
\node [left] at (0.3,1) {$b$};
\draw [fill] (4.5,1) circle (0.1);
\node [right] at (4.7,1) {$c$};
\draw [fill] (4.5,2.5) circle (0.1);
\node [right] at (4.7,2.5) {$d$};
\draw [fill] (2.0,1) circle (0.1);
\node [below] at (2.0,0.8) {$u$};
\draw [fill] (3.0,1) circle (0.1);
\node [below] at (2.9,0.8) {$w$};
\draw [fill] (5.5,0) circle (0.1);
\node [right] at (5.7,0) {$z$};
\node [above] at (2.6,-0.3) {$X_0$};
\draw [-] (0.5,1) -- (0.5,2.5);
\draw [-] (4.5,1) -- (4.5,2.5);
\draw [dotted,-] (0.5,1) -- (2.0,1);
\draw [dotted,-] (3.0,1) -- (4.5,1);
\draw [-] (2.0,1) -- (4.5,2.5);
\draw [-] (3.0,1) -- (4.5,2.5);
\draw [-] (3.0,1) -- (5.5,0);
\node at (6.5,1.7) {{\Large $\Rightarrow_F$}};
\draw [fill] (8,2.5) circle (0.1);
\node [left] at (7.8,2.5) {$a$};
\draw [fill] (8,1) circle (0.1);
\node [left] at (7.8,1) {$b$};
\draw [fill] (12.0,1) circle (0.1);
\node [below] at (12.0,0.8) {$c$};
\draw [fill] (12.0,2.5) circle (0.1);
\node [right] at (12.2,2.5) {$d$};
\draw [fill] (9.5,1) circle (0.1);
\node [below] at (9.5,0.8) {$u$};
\draw [fill] (10.5,1) circle (0.1);
\node [below] at (10.5,0.8) {$w$};
\draw [fill] (13.0,0) circle (0.1);
\node [right] at (13.2,0) {$z$};
\node [above] at (10.0,-0.3) {$X_1$};
\draw [-] (8,1) -- (8,2.5);
\draw [-] (12.0,1) -- (12.0,2.5);
\draw [dotted,-] (8,1) -- (9.5,1);
\draw [dotted,-] (10.5,1) -- (12.0,1);
\draw [-] (9.5,1) -- (10.5,1);
\draw [-] (10.5,1) -- (12.0,2.5);
\draw [-] (12.0,2.5) -- (13.0,0);
\node at (0,-2.2) {{\Large $\Rightarrow_F^\ast$}};
\draw [fill] (1.5,-1.5) circle (0.1);
\node [left] at (1.3,-1.5) {$a$};
\draw [fill] (1.5,-3) circle (0.1);
\node [left] at (1.3,-3) {$c$};
\draw [fill] (5.5,-3) circle (0.1);
\node [below] at (5.5,-3.2) {$b$};
\draw [fill] (5.5,-1.5) circle (0.1);
\node [right] at (5.7,-1.5) {$d$};
\draw [fill] (3.0,-3) circle (0.1);
\node [below] at (3.0,-3.2) {$w$};
\draw [fill] (4.0,-3) circle (0.1);
\node [below] at (4.0,-3.2) {$u$};
\draw [fill] (6.5,-4) circle (0.1);
\node [right] at (6.7,-4) {$z$};
\node [above] at (3.5,-4.3) {$X_{\kappa-1}$};
\draw [-] (1.5,-3) -- (1.5,-1.5);
\draw [-] (5.5,-3) -- (5.5,-1.5);
\draw [dotted,-] (1.5,-3) -- (3.0,-3);
\draw [dotted,-] (4.0,-3) -- (5.5,-3);
\draw [-] (3.0,-3) -- (4.0,-3);
\draw [-] (3.0,-3) -- (5.5,-1.5);
\draw [-] (5.5,-1.5) -- (6.5,-4);
\node at (8.0,-2.5) {{\Large $\Rightarrow_F$}};
\draw [fill] (9.5,-1.5) circle (0.1);
\node [left] at (9.3,-1.5) {$a$};
\draw [fill] (9.5,-3) circle (0.1);
\node [left] at (9.3,-3) {$c$};
\draw [fill] (13.5,-3) circle (0.1);
\node [below] at (13.5,-3.2) {$b$};
\draw [fill] (13.5,-1.5) circle (0.1);
\node [right] at (13.7,-1.5) {$d$};
\draw [fill] (11.0,-3) circle (0.1);
\node [below] at (11.0,-3.2) {$w$};
\draw [fill] (12.0,-3) circle (0.1);
\node [below] at (12.2,-3.2) {$u$};
\draw [fill] (14.0,-4.3) circle (0.1);
\node [right] at (14.2,-4.3) {$z$};
\node [above] at (11.5,-4.3) {$X_{\kappa}$};
\draw [-] (9.5,-3) -- (9.5,-1.5);
\draw [-] (13.5,-3) -- (13.5,-1.5);
\draw [dotted,-] (9.5,-3) -- (11.0,-3);
\draw [dotted,-] (12.0,-3) -- (13.5,-3);
\draw [-] (12.0,-3) -- (13.5,-1.5);
\draw [-] (11.0,-3) -- (13.5,-1.5);
\draw [-] (11.0,-3) -- (14.0,-4.3);
\end{tikzpicture}
\end{center}
\caption{Case 2: the simulation path}
\label{fig:Case2}
\end{figure}
After the detour flip, there is a path from
$b$ to $u$ in $X_1\setminus \{ a, d\}$, by choice of $u$,
and there is a path from $c$ to $w$ in $X_1\setminus \{ a,d\}$,
by choice of $w$.  It follows that there exists a $(b,c)$-path in $X_1$
which avoids $\{ a,d\}$: let $p$ be the lexicographically-least shortest such path.
Note that $p$ necessarily contains the edge $uw$.
(See $X_1$ in Figure~\ref{fig:Case2}: the path $p$ is indicated by the dotted lines
together with the edge from $u$ to $w$.)
We proceed as in Case~1, simulating
the switch $\Switch{a}{b}{c}{d}$ from $X_1$
using a long-flip with hub path $p$.  At the end of the long-flip the path
$p$ is reversed and we reach
the graph $X_{\kappa-1}$ shown at the bottom left of Figure~\ref{fig:Case2}.

For the final step of the simulation path, we reverse the detour flip
by performing a flip on the path $(u,w,d,z)$. This produces the
final graph $X_\kappa = Z'$ shown at the bottom right of Figure~\ref{fig:Case2},
completing the description of the simulation path in Case~2.

\subsection{Bounding the congestion}\label{sec:longflip-congestion}

\newcommand*\TR{TR}
\newcommand*\PR[1]{PR#1}
\newcommand*\D{D}
\newcommand*\DI{D$^{-1}$}

The previous section defines a set~$\SF$ of $(\mf,\msc)$-simulation paths.
First, we prove an upper bound on the maximum length $\ell(\SF)$
of a simulation path in $\SF$.

\begin{lemma}
If $r\geq 2$ then $\ell(\SF)\leq \nfrac{1}{2} n^2$.
\label{lem:lF}
\end{lemma}

\begin{proof}
Since hub paths are simple and do not include $a$ or $d$,
every hub path contains at most $n-2$ vertices.
Therefore, at most
$\frac12(n-2)(n-3)$ path-reversing flips will be required.
Furthermore, we may have up to $2r-1$ triangle-breaking flips,
and up to $2r-1$ triangle-restoring flips. Additionally, if we are
in Case 2 then we also have a detour flip and its inverse.
Therefore
\[
\ell(\SF) \leq \frac{(n-2)(n-3)}{2}+4r,
\]
and this is bounded above by $\nfrac{1}{2} n^2$ whenever
$r\geq 2$, since $n \geq 2r+1$.
\end{proof}

For the remainder of this subsection we work towards an upper bound on $B(\SF)$,
the maximum number of simulation paths containing a given flip.
A given flip~$t=(X,X')$ could perform one of a number of roles in a
simulation path for a connected switch $S = (Z,Z')$:
\begin{itemize}
\setlength\itemsep{-0.5mm}
\item a triangle-breaking flip;
\item a triangle-restoring flip;
\item a flip on the path reversal (that is, part of the long-flip);
\item in Case 2, the detour flip~$(u,w,d,z)$;
\item in Case 2, the inverse detour flip $(u,d,w,z)$.
\end{itemize}

We bound the number of simulation paths that could use~$t$ in each of these
roles separately, in Lemmas~\ref{lem:detour}--\ref{lem:pathreversal} below.
Summing over these bounds over all roles yields an upper bound for~$B(\SF)$.
In our exposition, we attempt to uniquely determine a given
simulation path $\flippath_{Z Z'}$ which contains $t$, by ``guessing''
from a number of alternatives.
Since the number of possible alternatives is precisely the number of
simulation paths $\flippath_{ZZ'}$ that use~$t$, this approach gives
the desired bound.

The simulation path is determined once we
have identified the switch $\Switch{a}{b}{c}{d}$ and either $Z$ or $Z'$
(since each of these can be obtained from the other using the switch or its
inverse).  Observe that in order to uniquely identify the switch
it is not enough to just identify the two edges $e_1$, $e_2$ which are switched
out (or to just identify the two edges that are switched in).
This follows since there are up to two possible switches on these
edges: if $e_1 = ab$ and $e_2=cd$ then we could switch in either the edges
$ac$, $bd$ or the edges $ad$, $bc$. We pay a factor of 2 to distinguish between these two
possibilities,
and will describe this as \emph{orienting} one of the switch edges.

Finally, suppose that the current flip $t = (X,X')$ switches out the edges $\alpha\beta$,
$\gamma\delta$ and replaces them by $\alpha\gamma$, $\beta\delta$.
If both $\alpha\delta$ and $\beta\gamma$ are present in $X$ (and hence,
both present in $X'$) then either one of these could have been the middle edge
of the 3-path chosen for the flip. In this case, if we wish to identify
$\beta$, say, then there are 4 possibilities. However, if we know a priori
that the edge $\alpha\delta$ is not present in $X$, say, then we can guess
$\beta$ from two possibilities.

\begin{lemma}
If $t=(X,X')$ is a detour flip then we can guess the connected switch $S=(Z,Z')$
from at most $8r^2 n$ possibilities.  The same bound holds
if $t$ is an inverse detour flip.
\label{lem:detour}
\end{lemma}

\begin{proof}
First suppose that $t$ is a detour flip.
This is the first flip in the simulation path, so $Z=X$.
\begin{itemize}
\setlength\itemsep{-0.5mm}
\item The flip is on the 3-path $(u,w,d,z)$, and $uz\not\in E(X)$ by construction,
so we can guess $d$ from 2 possibilities.
\item Next, we can guess $c\in N_X(d)$ from at most $2r$ possibilities.
\item
Finally, guess the edge $ab\in E(X)$ and orient it, from $2rn$ possibilities.
\end{itemize}
From the switch edges and $Z$ we can obtain $Z'$.
Hence we can guess the switch $S$ from at most $8r^2 n$ possibilities overall
when $t$ is a detour flip.
The same argument holds when $t$ is an inverse detour flip, replacing $X$ and $Z$ by $X'$ and $Z'$ and exchanging the roles of $b$ and $c$.
\end{proof}

\begin{lemma}
Suppose that $t=(X,X')$ is a triangle-breaking flip.  Then we can guess
the connected switch $S=(Z,Z')$ from at most $4(2r)^8 n$ possibilities.
The same bound holds if $t$ is a triangle-restoring flip.
\label{lem:trianglebreaking}
\end{lemma}

\begin{proof}
Suppose that $t$ is a triangle-breaking flip.
\begin{itemize}
\setlength\itemsep{-0.5mm}
\item One of the vertices involved in the flip $t$ is $b$,
so we can guess it from $4$ alternatives.
(If $t$ is not the first triangle-breaking flip then
by the no-shortcuts property, we can guess $b$ from 2 alternatives.
But in the first triangle-breaking flip $xp_2$ may be an edge of $Z$,
in which case there are 4 possibilities for $b$.)
\item We can guess~$a$ by adjacency to $b$, from $2r$ possibilities.
\item We can guess the edge $cd\in E(X)$ and orient it, from $2rn$ possibilities.
\end{itemize}
We now know the switch vertices, so
recovering~$Z$ will give us the rest of the path.  There may have been
some other triangle-breaking flips which took place before $t$, and if we
are in Case 2 then there was also a detour flip at the start of the simulation
path.  We continue our accounting below.
\begin{itemize}
\setlength\itemsep{-0.5mm}
\item
First suppose that we are in Case 1, and that $t$ is the $k$th triangle-breaking
flip which has been performed (so far).  Then $k\geq 1$. Let
$p'=(x,p_1,\ldots,p_{k},b)$ denote the path containing vertices which were
removed from the hub path by the first $k$ triangle-breaking flips.
Now $(a,p_1,x)$ is a 2-path in $X$, so we can guess $x$ and $p_1$ from
$2r(2r-1)\leq (2r)^2$ alternatives.
\item Next, $p_k$ belongs to $N_X(b)\setminus a$, so
we can guess $p_k$ from at most $2r$ alternatives.  (Here we allow the possibility that $p_1=p_k$, which covers the case that $t$
is the first triangle-breaking flip.)
Then we can reconstruct $p'$ using the
fact that $(p_1,\ldots, p_k)$ is the lexicographically-least
shortest $(p_1,p_k)$-path in $X\setminus \{a,d\}$,
by the recognisability property.
Using $p'$, we can reverse these $k$ triangle-breaking flips to obtain $Z$
from $X$.
\item
Finally, we must guess whether we are in Case 1 (1 possibility) or Case 2.
If we are in Case 2 then we must also reverse the detour flip before
we obtain $Z$.  Since $u$ and $w$ are neighbours of $d$ and $z$ is a neighbour of
$w$,
in this step we guess from at most $1 + (2r)^2\, (2r-1) < (2r)^3$ possibilities.
\end{itemize}
Overall, we can guess the switch $S$ from a total of at most
$  4 (2r)^8 n$
possibilities, when $t$ is a triangle-breaking flip.

The argument is similar when $t$ is a triangle-restoring flip:
\begin{itemize}
\setlength\itemsep{-0.5mm}
\item First, suppose that we are in Case 1.
Guess $b$, then guess $d$ from the neighbourhood of $b$,
then guess and orient the edge $ac\in E(X)$. This gives the
switch edges, guessed from at most $4\,(2r)^3$ alternatives.
\item Recover $p'=(x,p_1,\ldots, p_k,b)$ as in the triangle-breaking
case, from at most $(2r)^3$ alternatives.
Now we can $Z'$ uniquely by performing the remaining triangle-restoring
flips to move~$b$ past~$p_1$.
\item Again, there at most $(2r)^3$ possibilities for deciding
whether or not we are in Case 2 and if so, guessing the inverse
detour flip and applying it to produce $Z'$.
\end{itemize}
Hence there are at most $4(2r)^8 n$ possibilities for the connected switch
$S$ when $t$ is a triangle-restoring flip.
\end{proof}

\begin{lemma}
Suppose that $t=(X,X')$ is a flip on the path-reversal.
We can guess the connected switch $S=(Z,Z')$
from at most $4(2r)^8\, \left((2r)^2 - 1\right) n^3$ possibilities.
\label{lem:pathreversal}
\end{lemma}

\begin{proof}
Suppose that $t$ is a flip on the path-reversal with respect to the hub path
$p = (b,p_1,\ldots, p_{\nu-1},c)$, where $b=p_0$ and $c=p_{\nu}$.
Then $t$ is part of a $p_i$ run, for
some $i$. The path-reversal may have been preceded by some triangle-breaking
flips, and (if we are in Case 2) a detour flip.  We need to guess the
switch edges and reconstruct $Z$.

First suppose that we
are in Case 1. The typical situation is shown in Figure~\ref{fig:p-prime},
which shows that graph $X'$ produced by the flip $t$. The current bubble $p_i$
is shown as a diamond, and the flip $t$ has just been performed on the 3-path
$(p_j,p_i,p_{j+1},p_{j+2})$, moving $p_i$ past $p_{j+1}$.

\begin{figure}[ht!]
\begin{center}
\begin{tikzpicture}
\draw [fill] (1.0,0) circle (0.1);
\node [below] at (1.0,-0.2) {$x$};
\draw [fill] (2.0,0) circle (0.1);
\node [below] at (2.0,-0.2) {$p_1$};
\draw [fill] (3.0,0) circle (0.1);
\node [below] at (3.0,-0.2) {$p_2$};
\draw [dotted,-] (3.0,0) -- (4.5,0);
\draw [fill] (4.5,0) circle (0.1);
\node [left] at (4.4,-0.3) {$p_{k}$};
\draw [fill] (5.5,1.5) circle (0.1);
\node [right] at (5.7,1.5) {$a$};
\draw [fill] (5.5,0) circle (0.1);
\node [left] at (5.3,0.0) {$a'$};
\draw [dotted,-] (5.5,0) -- (7,0);
\draw [fill] (7,0) circle (0.1);
\node [above] at (7,0.2) {$p_j$};
\draw [fill] (8,0) circle (0.1);
\node [above] at (8,0.2) {$p_{j+1}$};
\draw [-] (9,0.1) -- (9.1,0) -- (9,-0.1) -- (8.9,0) -- (9,0.1); 
\node [above] at (9,0.2) {$p_{i}$};
\draw [fill] (10,0) circle (0.1);
\node [above] at (10,0.2) {$p_{j+2}$};
\draw [dotted,-] (10,0) -- (12,0);
\draw [fill] (12,0) circle (0.1);
\node [above] at (12,0.2) {$c$};
\draw [fill] (13,0) circle (0.1);
\node [above] at (13,0.2) {$c'$};
\draw [fill] (14.5,0) circle (0.1);
\node [above] at (14.5,0.2) {$p_{k+1}$};
\draw [fill] (15.5,0) circle (0.1);
\draw [dotted,-] (13,0) -- (14.5,0);
\draw [-] (14.5,0) -- (15.5,0);
\node [right] at (15.7,0) {$b$};
\draw [fill] (15.5,1.5) circle (0.1);
\node [right] at (15.7,1.5) {$d$};
\draw [-] (1.0,0) -- (3.0,0) -- (5.5,1.5) -- (3.0,0);
\draw [-] (2.0,0) -- (5.5,1.5) -- (5.5,0);
\draw [-] (4.5,0) -- (5.5,1.5);
\draw [-] (7,0) -- (8.9,0);
\draw [-] (9.1,0) -- (10,0);
\draw [-] (12,0) -- (13,0);
\draw [-] (15.5,0) -- (15.5,1.5);
\draw [rounded corners,-] (4.5,0) -- (4.5,-1.8) -- (15.5,-1.8) -- (15.5,0);
\draw [decoration={brace,mirror,amplitude=7}, decorate] (5.6,-0.3) -- (7.9,-0.3);
\node [below] at (6.75,-0.55) {$p[a' : p_{j+1}]$};
\draw [decoration={brace,mirror,amplitude=7}, decorate] (10.1,-0.3) -- (11.9,-0.3);
\node [below] at (11,-0.55) {$p[p_{j+2} : c]$};
\draw [decoration={brace,mirror,amplitude=7}, decorate] (13.1,-0.3) -- (14.4,-0.3);
\node [below] at (13.8,-0.55) {$\operatorname{rev} p[p_{k+1} : c']$};
\end{tikzpicture}
\end{center}
\caption{Reconstruction of $Z$ from $X'$, the flip edges and some guessed information.}
\label{fig:p-prime}
\end{figure}
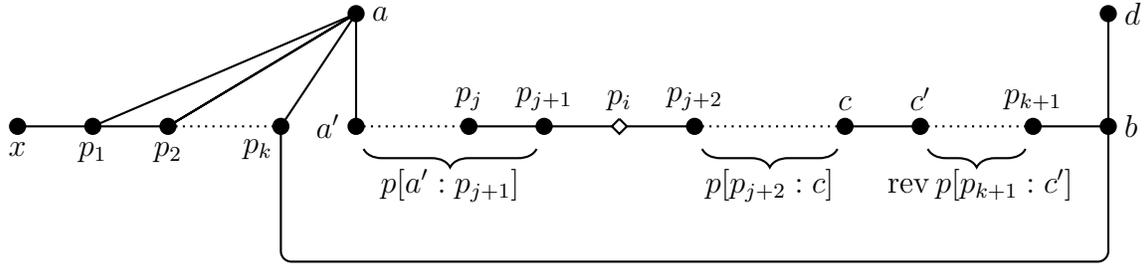

\begin{itemize}
\setlength\itemsep{-0.5mm}
\item The current bubble $p_i$ is one of the four vertices involved in the flip,
so it can be guessed from four possibilities.
\item Next we guess
$a$, $c$ and $d$ from~$n^3$ alternatives.
\item We recover~$b$ as
follows: if the current run is the~$b$-run then $b$ is the bubble;
otherwise, $b$ is already adjacent to $d$, and we can guess it from the
$2r$~neighbours of~$d$. This gives $2r+1$ possibilities for $b$.
\item
Vertices $a'$ and $c'$ can be guessed from at most
$2r$~alternatives each, by their adjacency to $a$ and $c$.
\end{itemize}
Now we know the switch $\Switch{a}{b}{c}{d}$,
and it remains to recover $Z$.  In order to do that, we must recover the
path $p$, reverse the path-reversal flips and then reverse the triangle-breaking
flips, if any.

Observe that having identified $p_i$, the identities of $p_{j+1}$ and
$p_{j+2}$ can be inferred. To see this, note that 
$p_ip_{j+2}$ is one of the two edges switched in by the flip $t$, 
which identifies $p_{j+2}$, while $p_{j+1}p_{j+2}$ is one of the two
edges switched out by the flip $t$, which identifies $p_{j+1}$.

Next, observe that $a'=p_{i+1}$, and so the subpath $p_{i+1}\cdots p_{j+1}$
is the lexicographically-least shortest $(p_{i+1},p_{j+1})$-path in
$X' \setminus \{ a,d\}$, by the recognisability property.
This holds since the flips made on the path-reversal so far, between
$Z$ and $X'$, have not altered this section of the original
path. 
Similarly, the subpath $p_{j+2}\cdots c$ can be recovered as
the lexicographically-least shortest $(p_{j+2},c)$-path in
$X'\setminus \{a,d\}$.

If the current bubble is $b$ then $c'=c$.
Otherwise, observe that $c'=p_{i-1}$, since $c'$ is the most recent
bubble that successfully completed its run.
We continue our accounting below.
\begin{itemize}
\setlength\itemsep{-0.5mm}
\item
Guess $p_{k+1}$ from $N_{X'}(b)\setminus d$, from at most $2r-1$ possibilities.
Reconstruct the subpath $p_{k+1}\cdots p_{i-1}$  as the lexicographically-least
shortest $(p_{k+1},p_{i-1})$-path in $X'\setminus \{a,d\}$.
\item
By now we have recovered all sections of the path between $a'$ and $b$ as shown
in Figure~\ref{fig:p-prime}, and we can reverse all path-reversing flips.
Now there may have been no triangle-restoring flips, in which case we have
constructed $Z$, or else we may need to
reverse the triangle-breaking flips as explained earlier,
after guessing the identities of $x$, $p_1$ and $p_{k-1}$ from at most
$(2r)^2(2r-1)$ possibilities.  This gives a factor of at most
$1 + (2r)^2(2r-1) \leq (2r)^3$.
\item
Finally, multiply by $1 + (2r)^2(2r-1)\leq (2r)^3$ to account for the fact that we
may be in Case 1 or Case 2, as explained in Lemma~\ref{lem:trianglebreaking}.
\end{itemize}
Overall, we have guessed the switch $S$ from at most
\[  4(2r+1)(2r-1)(2r)^8\, n^3 = 4 (2r)^{8}\left((2r)^2-1\right)  n^3
\]
possibilities, when $t$ is a flip on a path-reversal. 
\end{proof}

Now we apply the two-stage direct canonical path theorem to
our $(\mf,\msc)$-simulation paths.
Let $\bar{\rho}(\GSC)$ denote the congestion of some set of
canonical paths for the connected switch $\msc$.  We will define
such a set of canonical paths in Section~\ref{sec:disconnected} below.

\begin{lemma}
\label{lem:BF}
The $(\mf,\msc)$-simulation paths defined in Section~\ref{sec:long-flip}
define a set of canonical paths for $\mf$ with congestion
$\bar{\rho}(\GF)$ which satisfies
\[
\bar{\rho}(\GF) \leq  8\,(2r)^{11} n^4\, \bar{\rho}(\GSC).
 \]
\end{lemma}

\begin{proof}
Firstly, we bound the maximum number $B(\SF)$ of simulation
paths which contain a given flip~$t$, by adding together
the number of possibilities from
each of the five roles that $t$ may play.
Using Lemmas~\ref{lem:detour}--~\ref{lem:pathreversal},
we obtain
\begin{align*}
B(\SF) &\leq  4(2r)^{8} \left((2r)^2-1\right)  n^3 
 {} + 8 (2r)^8  n 
    {}           + 16 r^2  n   
         \\
	&\leq 4\left( (2r)^{10} - (2r)^8   + 2 (2r)^6 + 1\right) n^3\\
  &\leq 4(2r)^{10}\, n^3,
\end{align*}
since $n\geq 2r\geq 4$.
Next, we calculate the simulation gap $\DF = D(\mf,\msc)$.
Since $\mf$ and $\msc$ share the same
state space $\Omega_F$, and both have uniform stationary distribution,
\[\DF = \max_{\substack{uv\in E(\mf)\\zw\in E(\msc)}}\, 
          \frac{|\OF|\, \PSC(z,w)}{|\OF|\, \PF(u,v)} = 
          \frac{\PSC(z,w)}{\PF(u,v)}.
\]
Let $uv\in E(\mf)$ and $zw\in E(\msc)$. Then $u\neq v$, as self-loop transitions are not
included in $E(\mf)$ or $E(\msc)$. Hence
\[ \frac{1}{\PF(u,v)}\leq \frac{(2r)^3 n}{2} \quad \text{ and } \quad
 \PSC(z,w) = \frac{1}{3 a_{n,2r}} \leq \frac{2}{r^2n^2}\]
using (\ref{an2r}), since $n\geq 8$.
Therefore $\DF\leq 8r/n$.

Substituting these quantities into
Theorem~\ref{thm:twostagedirect}, using Lemma~\ref{lem:lF},
we obtain
\begin{align*}
\bar{\rho}(\GF)&\leq 
   D_{F}\, \ell(\Sigma_{F})\, B(\Sigma_{F})\, \bar{\rho}(\GSC)  \\
	&\leq  8\,(2r)^{11} n^4\, \bar{\rho}(\GSC),
\end{align*}
as claimed.
\end{proof}

\section{Disconnected graphs}
\label{sec:disconnected}

Next, we must define a set of $(\msc,\ms)$-simulation paths
and apply the two-stage direct canonical path method.
The process is similar to the approach introduced by Feder et
al.~\cite{saberi06switchflip}: the
chief difference is in the analysis.

We begin by defining a surjection $\h\colon \OS\to\OF$ such that
$\max_{G\in\OF}|\h^{-1}(G)|$ is polynomially bounded. Using this we construct
short paths in $\msc$ to correspond to single transitions in $\ms$.
In~\cite{saberi06switchflip} this step is responsible for a factor of
$\bigO{r^{34} n^{36}}$ in the overall mixing time bound. In contrast, by
using the
two-stage direct method we construct canonical paths with a relaxation of
only a factor of $120r n^3$.

First we must define the surjection $\h$. Let $G_S\in\OS$ and
let $\{ H_1,\ldots, H_k\}$ be the set of components in~$G_S$.
Let $v_i$ the vertex in~$H_i$ of highest label, for $i=1,\ldots, k$.
By relabelling the components if necessary, we assume that
\[ v_1 < v_2 < \cdots < v_k.\]
We call each $v_i$ an \emph{entry vertex}.
By construction, $v_k$ is the highest-labelled vertex in $G_S$.

In each component we will identify a \emph{bridge edge} $e_i=v_iv'_i$, where $v'_i$
is the neighbour of~$v_i$ with the highest label.
We call edge $v_i'$ an \emph{exit vertex}.
The graph~$G=h(G_S)$ is formed from $G_S$ by removing the bridge
edges~$e_1,\ldots, e_k$ and replacing them with \emph{chain edges}~$e'_1,\ldots, e_k'$,
where $e_i = v'_i v_{i+1}$ for $1\leq i\leq k-1$, and $e_k' = v_k' v_1$.
We call $e_k'$ the \emph{loopback edge} and we call this procedure \emph{chaining}.
(When discussing the chain, arithmetic on indices is performed cyclically;
that is, we identify $v_{k+1}$ with $v_1$ and $v'_{k+1}$ with $v'_1$.)

By construction we have $v_i' < v_i < v_{i+1}$ for $i=1,\ldots, k-1$.
Hence the ``head'' $v_{i+1}$ of the chain edge $e_i'$ is greater
than the ``tail'' $v_i'$, except for possibly the loopback edge:
$v_1$ may or may not be greater than $v_k'$.
We prove that $G=H(G_S)$ is connected in Lemma~\ref{lem:H-range-biconnected} below.

An example of a chained graph $G$ is shown in Figure~\ref{fig:chain} in the case $k=3$.
Here $G=H(G_S)$ for a graph $G_S$ with three components $H_1,H_2,H_3$.
The shaded black vertices are the entry vertices $v_1$, $v_2$ and $v_3$,
and
the white vertices are the exit vertices $v_1'$, $v_2'$ and $v_3'$.
After chaining, $H_1$ contains two cut edges (shown as dashed edges),
because $e_1 = v_1 v'_1$ was part of a cut cycle in $H_1$.
Similarly, $H_3$ contains one cut edge.  Note that these dashed edges are not chain
edges.
The loopback edge is $e'_3$ and $v_1 < v_2 < v_3$.

\begin{center}
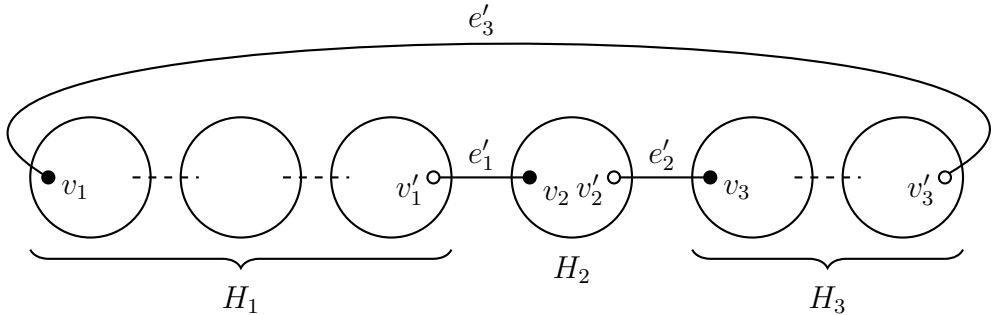
\begin{figure}[ht!]
\begin{center}
\begin{tikzpicture}[scale=0.8]
\draw (1,1) circle (1);
\draw (3.5,1) circle (1);
\draw (6,1) circle (1);
\draw (9,1) circle (1);
\draw (12,1) circle (1);
\draw (14.5,1) circle (1);
\draw [dashed,-] (1.7,1) -- (2.8,1);
\draw [dashed,-] (4.2,1) -- (5.3,1);
\draw [dashed,-] (12.7,1) -- (13.8,1);
\draw [fill] (0.3,1) circle (0.1);
\node [right] at (0.5,0.8) {$v_1$};
\draw  (6.7,1) circle (0.1);
\node [left] at (6.55,0.8) {$v_1'$};
\draw [fill] (8.3,1) circle (0.1);
\node [right] at (8.5,0.7) {$v_2$};
\draw  (9.7,1) circle (0.1);
\node [left] at (9.55,0.8) {$v_2'$};
\draw [fill] (11.3,1) circle (0.1);
\node [right] at (11.5,0.8) {$v_3$};
\draw  (15.2,1) circle (0.1);
\node [left] at (15.05,0.8) {$v_3'$};
\draw [-] (6.8,1) -- (8.3,1);
\node [above] at (7.5,1.1) {$e_1'$};
\draw [-] (9.8,1) -- (11.3,1);
\node [above] at (10.5,1.1) {$e_2'$};
\begin{scope}
\clip (-0.5,-1) rectangle (16.0,4.2);
\draw [-] (0.3,1) to [out=150,in=30] (15.27,1.07);
\end{scope}
\node [below] at (7.5,3.9) {$e_3'$};
\draw [decoration={ brace,mirror,amplitude=7}, decorate] (0,-0.2) -- (7,-0.2);
\node [below] at (3.5,-0.8) {$H_1$};
\node [below] at (9.0,-0.3) {$H_2$};
\draw [decoration={ brace,mirror,amplitude=7}, decorate] (11,-0.2) -- (15.5,-0.2);
\node [below] at (13.25,-0.8) {$H_3$};
\end{tikzpicture}
\caption{A connected graph $\h(G_S)\in\OF$ obtained by chaining a disconnected
graph $G_S\in\OS$ with three components $H_1, H_2, H_3$.}
\label{fig:chain}
\end{center}
\end{figure}
\end{center}

\begin{lemma}
\label{lem:H-range-biconnected}
Let $\h$ be the map defined above. Then
$h(G_S)$ is connected for all $G_S\in \OS$,
and $\h\colon \OS\to\OF$ is a surjection.
\end{lemma}

\begin{proof}
Let $G_S\in\OS$ and suppose $G=\h(G_S)$.
Lemma~\ref{lem:even-d}(ii) implies that each connected
component $H_i$ remains at least $1$-connected after removal of bridge edges.
The chain edges then connect these components in a ring, forming a connected graph.
Additionally, $\h$ preserves
degrees since each vertex $v_i$, $v'_i$ is adjacent to exactly one bridge edge
that is removed and exactly one chain edge that is added. Hence, $G\in\OF$.
Finally, observe that $h(G)=G$ whenever $G$ is connected, which proves
that $h$ is a surjection since $\OF\subset \OS$.
\end{proof}

\begin{lemma}
\label{lem:roughly-injective}
Let $\h\colon \OS\to\OF$ be defined as above.
Then
\[\max_{G\in\OF}\,|\h^{-1}(G)|\leq 2rn.\]
\end{lemma}

\begin{proof}
Let $G_F\in\Omega_F$ be given.
First we prove that given an edge $e = \{ v,w\}$ of $G_F$ and an orientation $(v,w)$
of $e$, there is a unique disconnected graph
$G_S\in h^{-1}(G_F)$ such that when the chaining procedure is performed
with input $G_S$ (producing $G_F$), the loopback edge $e_k'$ equals $e$
and $v_1=v$.
We recover $G_S$ from $G_F$ and the oriented edge $e$ using the \emph{unchaining
algorithm} which we now describe.

To identify $e'_1$ given $e'_{k}$ and $v_1$,
begin by searching the graph~$G_F\setminus e'_{k}$, starting from~$v_1$,
until a vertex is reached with a label greater than~$v_1$.  By construction,
this vertex must be $v_2$, as $v_1$ has the greatest label of any vertex in $H_1$.
The edge traversed to reach~$v_2$ must be $e'_1$.  Repeat this process
until the chain edges $e_1',\ldots, e_{k-1}'$ have been discovered.
Then $G_S$ is recovered by deleting all chain edges (including the loopback edge)
from $G_F$ and
replacing them with the bridge edges $e_1,\ldots, e_k$.

For future reference, note that this unchaining algorithm will still correctly
recover $G_S$ so long as $v'_i$ is any neighbour of $v_i$ in $G_S$
(that is, $v_i'$ does not need to be the neighbour with largest label).

Hence the number of choices for the oriented loopback edge  gives an upper bound
on $|h^{-1}(G_S)|-1$.  We use the naive estimate of~$2r(n-1)$
as an upper bound on this quantity, since the loopback edge could be
any oriented edge which does not originate in $v_k$ (the highest-labelled
vertex in $G_F$). Since $G_F$ is the only connected graph in $h^{-1}(G_F)$,
we obtain
$|h^{-1}(G_S)|\leq 2r(n-1) + 1 \leq 2rn$.
(It should be possible to obtain a tighter bound with a little more work, though
we do not attempt that here.)
\end{proof}

\subsection{Simulating general switches}\label{ss:simgen}

\newcommand*\BD{B$\Delta$}
\newcommand*\BR{BR}
\newcommand*\NLF{NS}
\newcommand*\DLF{DS}
\newcommand*\RLF{RS}
\newcommand*\DHK{DHK}
\newcommand*\RHK{RHK}

Now we must construct a set of $(\msc,\ms)$-simulation paths $\SSC$
with respect to the surjection $h$ defined above.
Given a switch $S = (W,W')\in E(\ms)$, with $W,W'\in\OS$,
we must simulate it by a path $\sigma_{WW'}$ in $\GSC$, consisting
only of connected switches. We then give an upper bound on the congestion
parameters for the set $\SSC$ of simulation paths, and apply
Theorem~\ref{thm:twostagedirect}.
The main difficulty is that switches in~$\ms$
may create or merge components, whereas connected switches in~$\msc$
cannot.

In this section we refer to transitions in $E(\ms)$ as \emph{general switches},
to distinguish them from connected switches. We say that a general switch $S=(W,W')$
is
\begin{itemize}
\setlength\itemsep{-0.5mm}
\item \emph{neutral} if $W$ and $W'$ have the same number of components;
\item \emph{disconnecting} if $W'$ has one more component than~$W$;
\item \emph{reconnecting} if~$W'$ has one fewer component than $W$.
\end{itemize}
The general switch $S=\Switch{a}{b}{c}{d}$ is disconnecting if $ab,cd$ form a
2-edge-cut
in the same component of $W$ and $ac,bd$ do not.
If $S$ is a reconnecting switch then $ab$ and $cd$ below to
distinct components of $W$. Finally, $S$ is a neutral switch if~$ab$
and $cd$ both belong to the same component of $W$, but do not form a
2-edge-cut,
or if $ab$ and $cd$ form a 2-edge-cut in a component of $W$,
and $ac$ and $bd$ form another 2-edge-cut in the same component.

To simulate the general switch $S=(W,W')$, we must define a path
\[\sigma_{W,W'} = (Z_0, Z_1,\dotsc,Z_q),\]
in $E(\msc)$,
where $Z_0=\h(W)$ and $Z_q=\h(W')$, such that $Z_i\in\OF$ for $i=0,\ldots, q$ and
$(Z_i,Z_{i+1})\in E(\msc)$ for $i = 0,\ldots, q-1$.
Here are the kinds of connected switch that we will need in our simulation
paths.
\begin{itemize}
\item  Suppose that a bridge edge $v_j v_j'$ of $W$ is one of the edges to be
removed by the switch~$S$.  This is a problem, since all bridge edges have been
deleted and replaced by chain edges during the chaining process used
to construct $Z_0=h(W)$.
To deal with this, we first perform
a connected switch called a \emph{bridge change} switch.
The bridge change switch reinstates the edge $e_je_j'$ (so that the desired
switch $S$ can be performed) and changes the choice of
exit vertex (previously $e_j'$) in~$H_j$.
Importantly, the resulting chain structure can still be unwound
using the algorithm in Lemma~\ref{lem:roughly-injective}.
Specifically, to reinstate the bridge
edge~$e_j=v_j v'_j$ in component~$H_j$, 
let $v^*_j\in N_{Z_0}(v_j)\setminus\set{v'_{j-1}}$ be the highest-labelled neighbour 
of~$v_j$ in~$Z_0[H_j]$
and
perform the connected switch~$\Switch{v_j}{v^*_j}{v'_j}{v_{j+1}}$. 
(Clearly $v_j^\ast\neq v_j'$ since the edge $v_jv_j'$ is absent in $Z_0$,
and $v_{j+1}v_j^\ast$ is not an edge of $Z_0$ due to the component structure
of $W$, so this connected switch is valid.)
The new chain edge is $v^*_j v_{j+1}$ and the new bridge edge is $v_j v_j^*$.
We denote the bridge change switch by \BD.
See Figure~\ref{fig:bridge-change}.
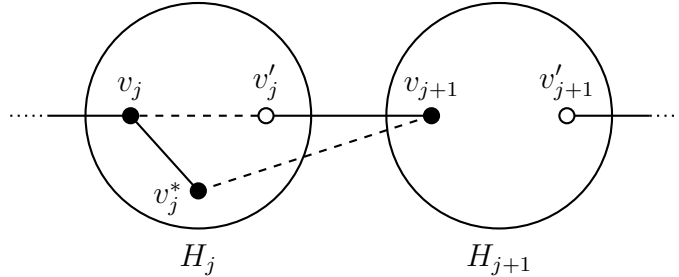
\begin{figure}[ht!]
\begin{center}
\begin{tikzpicture}[scale=1.0]
\draw (2.5,2.5) circle (1.5);
\node [below] at (2.5, 0.8) {$H_j$};
\draw (6.5,2.5) circle (1.5);
\node [below] at (6.5, 0.8) {$H_{j+1}$};
\draw [fill] (1.6,2.5) circle (0.1);
\node [above] at (1.6,2.7) {$v_j$};
\draw  (3.4,2.5) circle (0.1);
\node [above] at (3.4,2.7) {$v_j'$};
\draw [fill] (2.5,1.5) circle (0.1);
\node [left] at (2.3,1.4) {$v_j^*$};
\draw [fill] (5.6,2.5) circle (0.1);
\node [above] at (5.6,2.7) {$v_{j+1}$};
\draw  (7.4,2.5) circle (0.1);
\node [above] at (7.4,2.7) {$v_{j+1}'$};
\draw [-] (3.5,2.5) -- (5.6,2.5);
\draw [dashed,-] (3.3,2.5) -- (1.6,2.5);
\draw [-] (1.6,2.5) -- (2.5,1.5);
\draw [dashed,-] (2.5,1.5) -- (5.6,2.5);
\draw [-] (1.5,2.5) -- (0.5,2.5);
\draw [dotted,-] (0.0,2.5) -- (0.5,2.5);
\draw [-] (7.5,2.5) -- (8.5,2.5);
\draw [dotted,-] (8.5,2.5) -- (9.0,2.5);
\end{tikzpicture}
\caption{A bridge change in component $H_j$. }
\label{fig:bridge-change}
\end{center}
\end{figure}
\item We will need a connected switch to simulate the general switch $S$.
We denote this switch as \NLF, \DLF\ or \RLF\ if $S$ is  neutral, disconnecting
or reconnecting, respectively, and call it the \emph{neutral} (respectively,
\emph{disconnecting} or \emph{reconnecting}) \emph{simulation flip}.
It removes the same two edges that $S$ removes, and inserts the same two
edges that $S$ inserts.
\item  Let $Z$ be the current graph after the \DLF, \RLF\ or \NLF\ has been performed.
We say that component $H_j$ is \emph{healthy} in $Z$ if
there is no neighbour of the entry vertex $v_j$ in $Z$ with an label
higher than the exit vertex $v_j'$.
We must ensure that all components of $Z$ are healthy before performing the
unchaining algorithm, or else we lose the guarantee that $Z=h(W')$.

A component $H_j$ can only become unhealthy if some switch along the
simulation path involves the entry vertex $v_j$
and introduces a new neighbour
of $v_j$, say $v_j^*$, which has a higher label than the (current) exit vertex $v_j'$.
In this case, $v_j^*$ also has the highest label among all current neighbours
of $v_j$.
(In fact, the switch that makes $H_j$ unhealthy will be the switch
which simulates $S$, namely, the \DLF, \RLF\ or \NLF.)
Note, if there has been a bridge change in component $H_j$
at the start of the simulation path, then $v_j'$ might not be the exit vertex
originally produced in $H_j$ by the chaining procedure.
But this will not cause any problems, as we will see.

We can make $H_j$ healthy by performing the bridge change switch
$\Switch{v_j}{v^*_j}{v'_j}{v_{j+1}}$, precisely as described above.
(Now the main purpose of this switch is not to reinstate an edge, but to ensure
the correct vertex of $H_j$ becomes the exit vertex before unchaining.)
When this bridge change is performed at the end of the simulation path
we will refer to it as a \emph{bridge rectification} switch, denoted \BR.
\item If $S$ is a disconnecting switch then performing
this switch will produce a disconnected graph (and destroy the chain structure)
 by splitting $H_j$
into two components, one of which is not chained.
In order to avoid this, we perform a \emph{disconnected housekeeping switch},
denoted \DHK, before the disconnecting simulation flip (\DLF).
The \DHK\ is defined in detail in Section~\ref{ss:disconn} below.
\item Similarly, if $S$ is a reconnecting switch then performing this
switch destroys the chain structure, by merging $H_i$ and $H_j$ together, forming
one component which is incident with four chain edges instead of two.
To fix the chain structure we perform a \emph{reconnecting housekeeping switch},
denoted \RHK, after the
reconnecting simulation flip (\RLF).  More detail is given in Section~\ref{ss:conn} below.
\end{itemize}

We now define simulation paths for the three types of general switch~$S=(W,W')$.
The paths are summarised in Table~\ref{tab:switch-decompositions},
with optional steps denoted by square brackets.
Here ``optional'' means that these steps may or may not appear in the
simulation path for a general switch of that type. However, given
a particular switch $S$, the definition of the simulation path complete
determines whether or not these ``optional'' switches are required.

\renewcommand{\arraystretch}{1.2}
\begin{table}[ht!]
\begin{center}
\begin{tabular}{|lcccl|}
\hline
Neutral switch  && $\longrightarrow$ && [\BD] \NLF\ [\BR]\\
Disconnecting switch && $\longrightarrow$ && [\BD] \DHK\ \DLF\ [\BR] [\BR]\\
Reconnecting switch && $\longrightarrow$ && [\BD] [\BD] \RLF\ \RHK\ [\BR]\\
\hline
\end{tabular}
\caption{Simulation path for a general switch $S\in\OS$ using connected switches.}
\label{tab:switch-decompositions}
\end{center}
\end{table}

\subsubsection{The simulation path for a neutral switch}

Suppose that the general switch $S$ is a \emph{neutral switch}.
The simulation path for $S$ is formed using the following procedure.
\begin{itemize}
\setlength\itemsep{-0.5mm}
\item Let~$H_i$ be the component of $W$ that
contains all vertices involved in~$S$.
\item If the bridge edge~$v_iv'_i$ is switched away
by~$S$ then perform a bridge change.
\item After the bridge change, if any, perform a connected switch \NLF\ which switches
in and out the the edges specified by~$S$.
\item If a bridge edge is unhealthy after this switch then perform a bridge rectification.
\end{itemize}

\subsubsection{The simulation path for a disconnecting switch}\label{ss:disconn}

Suppose that the general switch $S$ is a \emph{disconnecting switch}.
The simulation path for $S$ is formed using the following procedure.
\begin{itemize}
\setlength\itemsep{-0.5mm}
\item Let~$H_j$ contain the
2-edge-cut~$ab,cd$ and let~$B_1,B_2$ be the subgraphs of~$H_j$ separated by this
cut. Without loss of generality, suppose that $B_1$ contains $v_j$.
\item
If a bridge change is required in $H_j$ then perform the bridge change.
In that case, $v_j v_j'$ is one of the edges forming the 2-edge-cut,
so all other neighbours of $v_j$ must lie in $B_1$.  Therefore the new exit
vertex $v_j^*$ also belongs to $B_1$.
\item
After the bridge change, if any, notice that switching in and out the
edges specified by $S$ would split $H_j$ into two components $B_1$ and $B_2$.
Of these, $B_1$ is already in the correct place in the chain, but $B_2$ is
not connected to the chain at all. This is not allowed, since we are
restricted to connected switches.  Therefore we must first perform the
\emph{disconnected housekeeping switch} (\DHK) to place $B_2$ into the
correct position in the chain.  This switch is described in more detail below.
\item After the \DHK\ has been performed, we
perform the disconnecting switch specified by $S$. (This is now safe, as the
\DHK\ ensures that the resulting graph is not disconnected.)
\item Up to two bridge edges may now be unhealthy. If so, then perform up to
two bridge rectification switches. For definiteness, if two \BR\ switches
are required, perform the one with the lower component index first.
\end{itemize}

It remains to specify the disconnected housekeeping switch \DHK.
See Figure~\ref{fig:dhk}.
Let $v_+$ be the vertex in $B_2$ with the highest label and let $v'_+$ be its
greatest neighbour. Then~$e_+ = v_+ v'_+$ must become the bridge edge of $B_2$.
Find $i$ such that $v_i < v_+ < v_{i+1}$, so $B_2$ belongs in the chain
between $H_i$ and $H_{i+1}$ in the chain. (Note that $v_k$ is the highest-labelled
vertex in the graph, so $v_+ \leq v_k$, but it is possible that $v_+ < v_1$,
in which case $B_2$ must become the first component in the chain and
$v_k' v_+$ must become the new loopback edge.)
Then \DHK\ is the switch~$\Switch{v_+}{v'_+}{v'_i}{v_{i+1}}$, which inserts $B_2$ into the correct
position in the chain as required.
\begin{figure}[ht!]
\begin{center}
\begin{tikzpicture}
\draw (0.5,4) circle (1.5);
\node [below] at (0.5, 2.2) {$H_i$};
\draw (4,4) circle (1.5);
\node [below] at (4, 2.2) {$H_{i+1}$};
\draw [fill] (-0.5,4.5) circle (0.1);
\node [below] at (-0.4,4.25) {$v_i$};
\draw  (1.5,4.5) circle (0.1);
\node [below] at (1.4,4.3) {$v_i'$};
\draw [fill] (3,4.5) circle (0.1);
\node [below] at (3.15,4.25) {$v_{i+1}$};
\draw (5,4.5) circle (0.1);
\node [below] at (4.8,4.4) {$v_{i+1}'$};
\draw (10,2.95) circle (2.0 and 2.25);
\draw [fill] (9,4.5) circle (0.1);
\node [below] at (9,4.3) {$v_j$};
\draw  (11,4.5) circle (0.1);
\node [below] at (11,4.3) {$v_j'$};
\draw [fill] (9,1.7) circle (0.1);
\node [below] at (9.2,1.45) {$v_+$};
\draw  (11,1.7) circle (0.1);
\node [below] at (10.7,1.55) {$v_+'$};
\draw [fill] (9.5,2.5) circle (0.1);
\node [left] at (9.3,2.5) {$b$};
\draw [fill] (9.5,3.5) circle (0.1);
\node [left] at (9.3,3.5) {$a$};
\draw [fill] (10.5,2.5) circle (0.1);
\node [right] at (10.7,2.5) {$d$};
\draw [fill] (10.5,3.5) circle (0.1);
\node [right] at (10.7,3.5) {$c$};
\draw [dashed,-] (8,3) -- (12,3);
\draw [-] (9.0,1.7) -- (10.9,1.7);
\draw [-] (9.5,2.5) -- (9.5,3.5);
\draw [-] (10.5,2.5) -- (10.5,3.5);
\draw [-] (1.6,4.5) -- (3.0,4.5);
\draw [-] (-0.5,4.5) -- (-1.5,4.5);
\draw [dotted,-] (-1.5,4.5) -- (-2,4.5);
\draw [-] (5.1,4.5) -- (6,4.5);
\draw [dotted,-] (6,4.5) -- (6.7,4.5);
\draw [-] (9,4.5) -- (8,4.5);
\draw [dotted,-] (8,4.5) -- (7.3,4.5);
\draw [-] (11.1,4.5) -- (12,4.5);
\draw [dotted,-] (12,4.5) -- (12.5,4.5);
\node [left] at (8.3, 1.4) {$H_j$};
\end{tikzpicture}
\caption{The situation just before a disconnecting housekeeping
switch (\DHK).}
\label{fig:dhk}
\end{center}
\end{figure}
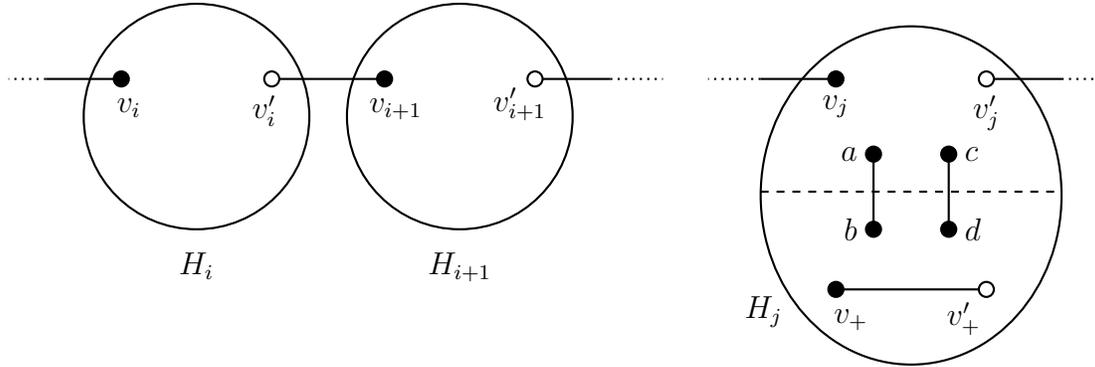

\subsubsection{The simulation path for a reconnecting switch}\label{ss:conn}

Suppose that the general switch $S$ is a reconnecting switch.
The simulation path for $S$ is defined using the following procedure.
\begin{itemize}
\setlength\itemsep{-0.5mm}
\item Suppose that
$ab\in H_i$ and $cd\in H_j$ where, without loss of generality $v_i < v_j$.
Up to two bridge changes may be necessary, as this switch involves two components.
If both are necessary, perform the bridge change in $H_i$ first.
\item
Then perform the \emph{reconnecting
switch} \RLF\ that switches in and out the edges specified by~$S$.
At this point, the component formed by linking $H_i$ and $H_j$
is incident with four chain edges, so the chain structure is not correct.  In
effect, this new component is occuping both position $i$ and position $j$
in the chain. Since $v_i < v_j$, vertex $v_j$ must be the entry vertex
of the new component, so the correct position for this component is between
$H_{j-1}$ and $H_{j+1}$.
\item The
\emph{reconnecting housekeeping}
switch~(\RHK)~$\Switch{v_i}{v'_{i-1}}{v'_i}{v_{i+1}}$ restores
the chain by placing the new component in position between $H_{j-1}$ and
$H_{j+1}$, linking $H_{i-1}$ directly to $H_{i+1}$ (which is now the next
component in the chain) and reinstating the bridge edge $v_iv_i'$.
Here $v_i'$ and $v_{i-1}'$ refer to endvertices of the \emph{current} chain
edges, \emph{after} the at most two bridge changes have been performed.
(These may differ from the endvertices of the original chain edges, before
the bridge rectification steps.)
\item Finally, at most one bridge
rectification may be required to  ensure that the single bridge edge $v_jv_j'$
is healthy.
\end{itemize}

In the special case when the reconnecting switch $S=(W,W')$ results in a connected
graph $W'$, we must have $i=1$ and $j=2$: here the reconnecting housekeeping switch
reinstates \emph{both} bridge edges, and no bridge rectification is required.

\subsection{Bounding the congestion}

Let us calculate upper bounds for ~$B(\SSC)$ and~$\ell(\SSC)$.
From Table~\ref{tab:switch-decompositions}, the longest possible simulation path consists
of five connected switches, so
\begin{equation}
\label{ellSSC}
 \ell(\SSC)=5.
\end{equation}
We spend the rest of this subsection bounding $B(\SSC)$.

Each connected switch~$t=(Z,Z')$ can play a number of roles
in a simulation path. We consider each kind of general switch
separately (neutral, disconnecting, reconnecting).
As in Section~\ref{sec:longflip-congestion}, the number of simulation paths
(of a given kind) containing $t$ is equal to the amount of extra
information required to identify a general switch $S$ of this kind
uniquely, given $t$ (and assuming that $t$ belongs to the simulation
path for $S$): we speak of ``guessing'' this information.
Again, it is enough to identify one of $W$ or $W'$, together
with the edges switched in/out by $S$, and then the other graph ($W'$ or $W$)
can be inferred.

\subsubsection{Neutral switches}\label{neu}

A connected switch~$t=(Z,Z')$ in~$\msc$ could take on three possible roles when
simulating a neutral switch~$S=(W,W')$.   We consider each separately below.

\bigskip

\noindent({\bf \BD})\ First, suppose that the connected switch $t$ is a bridge change in a
simulation path for a neutral (general) switch $S$.
\begin{itemize}
\setlength\itemsep{-0.5mm}
\item  Since $t$ is a bridge change, one of the
edges which is switched in by $t$ is also an edge which
is switched out by $S$. Guess which one, from $2$~alternatives.
\item
The other edge which is switched out by $S$ is an edge of $Z'$, so this
edge and its orientation can be guessed
from fewer than $2rn$ alternatives.
\item Also guess whether $W$ is connected and if not, guess
the oriented loopback edge from at most $2rn$ alternatives
(as in the proof of Lemma~\ref{lem:roughly-injective}).
If $W$ is connected then $W=Z$.
Otherwise, we must apply the unchaining algorithm from Lemma~\ref{lem:roughly-injective}
to the graph~$Z$ with respect to the chosen oriented
loopback edge. This produces the graph $W$, and then using the switch edges
we can determine $W'$.
\end{itemize}
In all we guessed $S$ from $2\cdot (2rn)^2  = 8r^2n^2$ alternatives.

\bigskip

\noindent ({\bf \NLF})\
Next, suppose that the connected switch $t$ is a neutral simulating switch
for a neutral (general) switch $S$.
\begin{itemize}
\setlength\itemsep{-0.5mm}
\item
The edges switched in/out by $t$ are precisely the edges switched in/out by $S$.
This gives us the vertices $\{a,b,c,d\}$ involved in the switch.
Guess the (current) oriented loopback edge in $Z$ from at most $2rn$ alternatives,
and perform the unchaining algorithm to reveal the chain structure of $Z$.
In particular, we now know the entry vertex and (current) exit vertex of
every component.
\item If $t$ was not preceded by a bridge change then $Z=Z_0$.
Otherwise, a bridge change has been performed before $t$,
which implies that $Z=Z_1$. In this case, one of the edges switched out by
$t$ was originally a bridge edge in $W$,
and hence is incident with $v_i$ for some component $H_i$ of $W$.
Since $S$ is a neutral switch, all vertices involved in the switch belong
to the same component $H_i$, and hence $v_i$ has the largest label among
$a,b,c,d$.  Without loss of generality, if $a=v_i$ then $b=v_i'$,
while our knowledge of the chain structure of $Z$ gives us the
current exit vertex $v_i^\ast$. Hence we can uniquely reverse
the bridge change switch to determine $Z_0$.
Therefore we can find $Z_0$ by guessing between just two alternatives, namely,
whether or not a bridge change was performed.

From $Z_0$ we obtain $W$ by deleting the chain edges and reinstating the
bridge edges (taking care to adjust the exit vertex of $H_i$ if a bridge change
has been performed).
Since we know $W$ and the switch edges, we can determine $W'$.
\end{itemize}
We have guessed $S$ from $4rn$ alternatives, when $t$ is a neutral simulating switch.
(The above analysis holds even if the switch $S$ replaces
one 2-edge-cut in $W$ by another.)

\bigskip

\noindent({\bf \BR})\
Finally, suppose that $t$ is a bridge rectification switch in a simulation
path for a neutral (general) switch $S$.
Then $t$ is the final switch in the simulation path.
\begin{itemize}
\setlength\itemsep{-0.5mm}
\item
Guess the oriented loopback edge in $Z'$, from at most $2rn$
alternatives, and
perform the unchaining algorithm on~$Z'$ to produce $W'$.
(Again, this includes the possibility that $W'$ is connected, in which case $W'=Z'$.)
\item
Since $t$ is a bridge rectification, the switch $S$ switches in an edge which is
incident with the entry vertex $v_i$ of $H_i$, giving it a new neighbour
with a higher label than the current exit vertex in $H_i$.
This edge is one of the two edges switched out by $t$: guess which one, from 2 alternatives.
\item
The other edge which is switched in by $S$, together with its orientation, can be guessed
from the (oriented) edges of $Z$: at most~$2rn$ possibilities.
Now $W$ can be inferred from $W'$ and the switch edges.
\end{itemize}
We have guessed $S$ from at most $2rn\cdot 2\cdot 2rn = 8r^2 n^2$
alternatives.

\bigskip

Adding the contributions from these three roles, a connected switch $t$ may be part of at most
\begin{equation}
\label{neutral}
8r^2n^2  + 4rn + 8r^2n^2 \leq  (8r^2+1)n^2
\end{equation}
simulation paths for neutral general switches.

\bigskip

Before we move on to consider disconnecting and reconnecting switches,
we make a couple of comments. The analysis of the neutral simulation switch
(\NLF) shows that once the edges of the general switch $S$ have
been guessed, we have enough information
to decide whether a bridge rectification switch is required and if so, its
specification.  Since these bridge rectification switches do not affect the number of choices
for $S$ in each subcase below, we only mention them when $t$ itself is playing
the role of a bridge rectification, and otherwise we omit them from our accounting.

In every situation we must guess the oriented loopback edge, so that we can perform the
unchaining algorithm. In some cases we must also
be careful to specify which graph is given as input to the
unchaining algorithm. In particular, we must not perform the unchaining algorithm
immediately after a disconnecting housekeeping switch or immediately before
a reconnecting housekeeping switch, since in those graphs the chaining structure
has been temporarily compromised.

\subsubsection{Disconnecting switches}\label{s:discon}

A connected switch~$t=(Z,Z')$ in~$\msc$ could take on five possible roles when
simulating a disconnecting switch~$S=(W,W')$.   We consider each separately below.

\bigskip

\noindent ({\bf \BD})\ If $t$ is a bridge change then we can guess $S$ from at most
$8r^2n^2$ alternatives, as explained in Section~\ref{neu}.

\bigskip

\noindent ({\bf \DHK})\ Next, suppose that $t$ is a disconnecting housekeeping switch.
\begin{itemize}
\setlength\itemsep{-0.5mm}
\item Guess the (current) oriented loopback edge as an edge of $Z$ from
at most $2rn$ possibilities.
Perform the unchaining algorithm on $Z$ to reveal the chain structure of $Z$.
Again, this tells us the entry and (current) exit vertex of every component of $W$.
\item Decide whether or not a bridge change was performed before $t$,
from 2 possibilities.
As described in Section~\ref{neu}, we can uniquely
determine the bridge change from the chain structure of $Z$, so in both cases
we can find $Z_0$ (note that $Z=Z_0$ if there was no bridge change,
and $Z=Z_1$ if a bridge change was performed before $t$).  By reversing the bridge
change if necessary we obtain $Z_0$, and then from our knowledge of the chain
structure we can determine $W$.
(If a bridge change was performed then we must use the original exit vertex for that
component: this vertex is revealed by reversing the bridge change.)

Since $t$ is a disconnected housekeeping switch, it switches out a chain edge
$v'_jv_{j+1}$, for some $j$, and
an edge $v_+v'_+$ which is entirely contained within a component $H_i$, for some $i\neq j$.
Hence we can distinguish between these two edges of $t$ and identify which is the chain
edge and which is $v_+v_+'$, also determining $i$ and $j$.
Furthermore, we can identify $v_+$ uniquely as it has a higher label than $v_+'$.
\item
Let $v_i$ be the entry vertex in $H_i$.
A disconnected housekeeping switch is only performed when the disconnected simulation
flip (\DLF) representing $S$ would switch out
the two edges of a 2-edge-cut separating $v_+$ and $v_i$, say.
By Lemma~\ref{lem:even-d}(iv), we can guess which 2-edge-cut is switched out by $S$
from at most~$n^2/(15 r^2)$ possibilities.
The orientation of the switch $S$ is uniquely determined, since there
is only one possibility which would disconnect $H_i$.
Now we know $W$ and the switch edges, we can determine $W'$.
\end{itemize}
We have guessed $S$ from at most $2rn\cdot 2\cdot n^2/(15 r^2) = 4n^3/(15r)$
possibilities.

\bigskip

\noindent ({\bf \DLF})\   Now suppose that $t$ is a disconnecting simulation switch.
The edges switched in/out by $t$ are precisely the edges switched in/out by $S$,
so we do not need to guess these.
Note that $Z$ does not have a valid chain structure, but $Z'$ does.
\begin{itemize}
\setlength\itemsep{-0.5mm}
\item
We know that $W'$ is not connected, so we guess the oriented
loopback edge from the (oriented) edges of $Z'$: at most $2rn$ possibilities.
Perform the unchaining algorithm on $Z'$, revealing the chain structure of $Z'$.
Suppose that $H_i$ and $H_j$ are the
two components of $Z'$ which contain an edge which was switched in by $t$,
and suppose that $i < j$. Then the disconnecting housekeeping switch (\DHK) is
uniquely determined: it is the 
switch~$\Switch{v_i}{v'_i}{v'_{i-1}}{v_{i+1}}$.
\item
Guess whether or not a bridge change was performed before the disconnecting housekeeping
switch: if so, it is uniquely specified (as explained earlier), so we guess from 2
possibilities.  Reverse this bridge change to obtain $Z_0$ and hence find $W$.
From $W$ and the switch edges we can determine $W'$.
\end{itemize}
We have guessed the general switch $S$ from at most $4rn$ possibilities.

\bigskip

\noindent ({\bf The first \BR})\
Now suppose that $t$ is the first bridge rectification switch in the
simulation path for a disconnecting (general) switch $S$.  There may or may
not be a second bridge rectification switch after $t$.
The calculations are similar to the bridge rectification case in Section~\ref{neu}.
\begin{itemize}
\setlength\itemsep{-0.5mm}
\item Guess the oriented loopback edge in $Z'$, from at most $2rn$ alternatives,
and perform
the unchaining algorithm to reveal the chain structure of $Z'$.   (Note, we know that
$W'$ is disconnected when $S$ is a disconnecting switch.)
\item Since $t$ is a bridge rectification switch, one of the edges switched in by
$S$ is $v_iv_i^\ast$, where $v_i$ is the entry vertex of some component $H_i$ and
$v_i^\ast$ has a higher label than the current exit vertex $v_i'$ in that component.
This edge is then switched out by $t$.
(It is possible that this condition is satisfied by both edges of $S$, but if so,
we know that the first bridge rectification will act on the component $H_i$ with the
lower index first.)
Guess which edge of $t$ is $v_iv_i^\ast$, for a factor of 2.
\item Guess the other
edge switched out by $S$, and the orientation of the switch, from at most $2rn$ possibilities.
(The other edge switched out by $S$ must lie in a distinct component from $H_i$,
but for an upper bound we ignore this.)
Whether or not a second bridge change will be needed after $t$ is completely determined
by the switch edges and the chain structure, as described in Section~\ref{neu}.
If a second bridge rectification is needed, perform it to produce $Z_{q}$. From $Z_q$
we may determine $W'$ using our knowledge of the chain structure of $Z'$
(and adjusting the exit vertex of the component in which the second bridge rectification
was performed, if necessary).
From $W'$ and the switch edges we may deduce $W$.
\end{itemize}
We have guessed $S$ from at most $2rn\cdot 2\cdot 2rn = 8 r^2n^2$ alternatives.

\bigskip

\noindent ({\bf The second \BR})\
Here we assume that one bridge rectification has been already performed to produce $Z$.
The connected switch $t$ is the last one in the simulation path.
\begin{itemize}
\setlength\itemsep{-0.5mm}
\item Again, we guess the oriented loopback edge for $Z'$, from at most $2rn$ possibilities.  Apply the unchaining algorithm to produce $W'$.
\item Choose which edge switched out by $t$ is an edge switched in by $S$,
out of 2 possibilities. The endvertex of this edge with the higher label
is $v_i$, for some $i$.
\item Since this is the second bridge rectification, the other edge
switched in by $S$ was $v_jv_j^*$ for some other component $H_j$,
and this edge has been made into the bridge edge in $H_j$, using the first bridge
rectification.
Therefore we can identify and orient this edge once we guess $H_j$, from at most $n/r$
possibilities (using Lemma~\ref{lem:even-d}(iii)).
This determines the edges switched in by $S$.
From $W'$ and the switch edges, we can determine $W$.
\end{itemize}
We have
guessed $S$ from at most $2rn\cdot 2\cdot n/r = 4n^2$ possibilities.

\bigskip

Summing over these five roles, a connected switch $t$ can be included in at most
\begin{equation}
\label{disconnecting}
8r^2  n^2 +  4n^3/(15r) + 4rn + 8r^2n^2 + 4n^2   \leq 9rn^3
\end{equation}
simulation paths for disconnecting general switches.
(The upper bound follows since $2\leq r\leq n/2$.)

\subsubsection{Reconnecting switches}\label{recon}

A connected switch~$t=(Z,Z')$ in~$\msc$ could take on give possible roles when
simulation a reconnecting (general) switch~$S=(W,W')$. We consider each
separately below.

\bigskip

\noindent ({\bf The first~\BD})\
When $t$ is the first bridge change we can guess $S$ from at most
$8 r^2 n^2$ alternatives,
using arguments very similar to those given in Section~\ref{neu}.

\bigskip

\noindent ({\bf The second~\BD})\
Now suppose that $t$ is the second of two bridge changes in the simulation
path for some reconnecting (general) switch $S$.
\begin{itemize}
\setlength\itemsep{-0.5mm}
\item Guess the oriented loopback edge in $Z$,
from at most $2rn$ possibilities.
 Perform the unchaining algorithm to reveal the chain structure of $Z$.
\item
Since $t$ is a bridge change, one of the edges switched in by $t$ is an edge which will
be switched out by $S$.
Choose which, from 2 possibilities. This edge is incident with the entry
vertex $v_j$ of some component $H_j$.
\item
Since $t$ is the second bridge change,
the other edge which will be switched out by $S$ must be incident with
the entry vertex of some other component $H_i$, with $i<j$.
(It has been put back into $Z$ by the first bridge change.) By Lemma~\ref{lem:even-d}(iii)
there are at most $n/(2r)$ components, so we choose one to determine $v_i$.
\item
Choose a non-chain edge incident with $v_i$ and orient it, from $2(2r-1)$ possibilities.
This determines the edges switched in and out by $S$. From these edges and $W$,
we may determine $W'$.
\end{itemize}
We have guessed $S$ from at most
$2rn\cdot 2\cdot n/(2r)\cdot 2(2r-1)  < 8r n^2$ possibilities.

\bigskip

\noindent ({\bf \RLF})\
Suppose that $t$ is a reconnecting simulation switch.
The edges switched in/out by $t$ are precisely the edges switched in/out by $S$.
\begin{itemize}
\setlength\itemsep{-0.5mm}
\item
Guess the oriented loopback edge for $Z$ from at most $2rn$ possibilities.
Perform the unchaining algorithm on $Z$ to reveal the chain structure of $Z$.
\item
Up to two bridge changes may have been performed before $t$, corresponding
to the up to two edges switched out by $S$ which are incident with entry vertices
in some component of $Z$. If two bridge changes have been performed then their order
is uniquely determined. As argued in Section~\ref{neu}, we must just decide whether
or not a bridge change has been performed in each component, and then the rest is specified,
so there are $2^2=4$ possiblities for the bridge changes before $t$, including
the possibility that there were none.
Once the bridge change switches are known,
they can be reversed, which reveals $W$. Together with the switch edges,
this determines $W'$ and hence $S$.
\end{itemize}
We have guessed $S$ from at most $8rn$ possibilities.

\bigskip

\noindent ({\bf \RHK})\
Suppose that $t$ is a reconnecting housekeeping switch.
Then $Z'$ has a valid chain structure, though $Z$ does not.
\begin{itemize}
\setlength\itemsep{-0.5mm}
\item
Guess the oriented loopback edge for $Z'$ from at most $2rn$ possibilities. Perform the
unchaining algorithm on $Z'$ to reveal the chain structure of $Z'$.
\item
The switch $t$ involves entry vertices $v_i$ and $v_j$, where $i<j$.
Since $t$ is a reconnecting housekeeping switch, the edges switched in by $S$ form a
2-edge-cut in $Z'$
which separate $v_i$ and $v_j$. By Lemma~\ref{lem:even-d}(iv)
we can guess this 2-edge-cut from at most $n^2/(15 r^2)$ possibilities.
This specifies the edges switched out by $S$.
The orientation of the switch is determined by the
fact that one edge switched out by $S$ is contained in $H_i$ and the
other is contained in $H_j$.
\item
Finally, we must decide whether or not a bridge rectification is needed
to produce $W'$, giving 2 possibilities. If a bridge rectification is needed
then it is uniquely determined: perform it to obtain $W'$, and then $W$ can
be obtained using the switch edges.
\end{itemize}
We have guessed $S$ from a total of at most
$2rn\cdot n^2/(15 r^2)\cdot 2 = 4n^3/(15 r)$ possibilities.

\bigskip

\noindent ({\bf \BR})\
Suppose that $t$ is a bridge rectification for a reconnecting (general) switch $S$.
Then $t$ is the final switch in the simulation path. In particular, $W'$ is 
disconnected, since otherwise the reconnecting housekeeping switch produces
$W'$, and no \BR\ or unchaining is necessary.
\begin{itemize}
\setlength\itemsep{-0.5mm}
\item Guess the oriented loopback edge and perform the unchaining
algorithm in $Z'$ to produce $W'$, from at most $2rn$ alternatives.
\item One of the edges switched in by $S$ is an edge which was switched out
by $t$: choose one, from 2 possibilities, and call it $e$.
\item
Let $\widehat{e}$ be the other edge which is switched in by $S$. 
Recall that the edges removed by the \RHK\ are
$e_{i-1}=v_iv_{i-1}'$ and $e_i=v_i'v_{i+1}$, where $e_i'$, $e_{i-1}'$ 
are the chain edges in the graph $Z_{q}$ obtained \emph{after}
the (at most two) bridge changes were performed. 
(See Section~\ref{ss:conn}.)
Let $Z_{q+1}$ be the result of performing the \RLF\ simulating $S$, starting from $Z_q$.
Since the edges switched in by the \RLF\ must be disjoint from
the edges of $Z_q$, we conclude that $\widehat{e}\not\in\{ e_{i-1}',\, e_i'\}$:
the three edges $\widehat{e}$, $e_{i-1}'$, $e_i'$ are distinct and all
belong to $Z_{q+1}$.
Next, the \RHK\ $(Z_{q+1},Z_{q+2})$ removes precisely two edges, namely $e_{i-1}'$,
$e_{i}'$, and hence $\widehat{e}$ is
still present in $Z=Z_{q+2}$.  Therefore we can guess and orient $\widehat{e}$
from among the $2rn$ oriented edges of $Z$. 
(In fact, the two edges
switched in by $S$ form a 2-edge-cut in $Z$, but we do not use that fact here.)
Now we know $W'$ and the switch edges, we can determine $W$.
\end{itemize}
We have guessed $S$ from at most $2rn\cdot 2\cdot 2rn = 8r^2 n^2$
possibilities.

\bigskip

Summing the contribution from these five roles, a connected switch $t$ can be included
in at most
\begin{equation}
\label{reconnecting}
8 r^2 n^2 + 8 r n^2 + 8rn +  4n^3/(15 r) + 8 r^2 n^2 \leq  11 r n^3
\end{equation}
simulation paths for reconnecting general switches.
(The upper bound follows since $2\leq r\leq n/2$.)

\subsection{Completing the analysis}

Now we apply the two-stage direct canonical path construction
to our set of $(\msc,\ms)$-simulation paths.
Here $\bar{\rho}(\GS)$ denotes the congestion of any set
of canonical paths, or multicommodity flow, for the switch chain.
We will use the bound given in~\cite{cooper05sampling}.

\begin{lemma}
The $(\msc,\ms)$-simulation paths defined in Section~\ref{ss:simgen}
define a set of canonical paths for $\msc$ with congestion
$\bar{\rho}(\GSC)$ which satisfies
\[
\bar{\rho}(\GSC) \leq 120\, r n^3\, \bar{\rho}(\GS).
\]
\label{lem:BSSC}
\end{lemma}

\begin{proof}
First, by adding together (\ref{neutral}) -- (\ref{reconnecting}),
the maximum number $B(\SSC)$
of simulation paths containing a given connected switch satisfies
\[
B(\SSC)	\leq (8r^2+1)n^2 + 9r n^3 + 11 rn^3
	\leq 24 r n^3,
\]
using the fact that $r\geq 2$ and $n\geq 2r+1\geq 5$.
Next we calculate the simulation gap $\DS = D(\msc,\ms)$.
Recall that for all $uv\in E(\msc)$ and $zw\in E(\ms)$ we have $u\neq v$ and $z\neq w$,
and hence
\[ \PSC(u,v) = \PS(z,w) = \frac{1}{3 a_{n,r}},\]
by definition of both chains.
The state space of $\ms$ is $\OS$ and the state space of $\msc$ is
$\OF$, with $\OF\subseteq \OF$.
Since both $\ms$ and $\msc$ have uniform stationary distribution, 
\[\DS = \max_{\substack{uv\in E(\msc)\\zw\in E(\ms)}}
         \frac{|\OF|\, \PS(z,w)}{|\OS|\, \PSC(u,v)} \leq
            \max_{\substack{uv\in E(\msc)\\ zw\in E(\ms)}}\,
         \frac{\PS(z,w)}{\PSC(u,v)} =
   1.\]
Substituting these values and (\ref{ellSSC}) into Theorem~\ref{thm:twostagedirect}
gives
\begin{align*}
\bar{\rho}(\GSC)&\leq 
   D_{SC}\, \ell(\Sigma_{SC})\, B(\Sigma_{SC})\, \bar{\rho}(\GS) \\
	&\leq 120\, r n^3\, \bar{\rho}(\GS),
\end{align*}
as required.
\end{proof}

Finally, we may prove Theorem~\ref{thm:flip-mixing-time}.

\begin{proof}[Proof of Theorem~\emph{\ref{thm:flip-mixing-time}}]\
Combining the bounds of Lemma~\ref{lem:BF} and Lemma~\ref{lem:BSSC} gives
\begin{equation}
\label{relaxation}
\frac{\bar{\rho}(\GF)}{\bar{\rho}(\GS)} \leq
	 480\,(2r)^{12} n^7 .
\end{equation}
The upper bound on the mixing
time of the switch chain given in~\cite{cooper05sampling,corrigendum}
has (an upper bound on) $\bar{\rho}(\GS)$ as a factor.
Therefore, multiplying this bound
by the right hand side of (\ref{relaxation}), we conclude that
the mixing time of the flip chain is at most
\[
   480\, (2r)^{35}\, n^{15}\, \left(2rn\log(2rn) + \log(\varepsilon^{-1})\right),
\]
completing the proof of Theorem~\ref{thm:flip-mixing-time}.
(This uses the fact that $|\Omega_S| \leq (rn)^{rn}$, which follows
from known asymptotic enumeration results~\cite{bendercanfield}.)
\end{proof}

It is unlikely that the bound of Theorem~\ref{thm:flip-mixing-time}
is tight.  Experimental evidence was presented in an earlier version
of this work~\cite{CDH} which provides support for the conjecture that
the true mixing time for the flip chain is~$\bigO{n\log n}$,
when $r$ is constant.
Perhaps $O(rn\log n)$ is a reasonable conjecture when $r$ grows with $n$.

\appendix
\section{Maximising 2-edge-cuts separating two vertices}

We now present the deferred proof of Lemma~\ref{lem:even-d}(iv).
Our aim here is to present a simple upper bound on $\Lambda(u,v)$
which holds for all relevant values of $n$ and $r$, though our
proof establishes tighter bounds than the one given in the statement
of Lemma~\ref{lem:even-d}(iv).

\begin{proof}[Proof of Lemma~\emph{\ref{lem:even-d}(iv)}]
Let $\Lambda(u,v)$ denote the number of 2-edge-cuts in $G$ which
separate $u$ and $v$.
By considering the connected component of $G$ which contains $u$,
if necessary, we may assume that $G$ is connected.
Consider the binary relation $\sim$ on $V$ defined by
$w_1\sim w_2$ if and only if there are 
at least three edge-disjoint paths
between $w_1$ and $w_2$ in $G$. Then $\sim$ is an equivalence
relation~\cite{tutte}
which partitions the vertex set $V$ into equivalence classes $U_0,\ldots, U_k$.
For $j=0,\ldots, k$ let $H_j = G[U_j]$ be the subgraph of $G$ induced
by $U_j$. Each $H_j$ is either a maximal 3-edge-connected induced
subgraph of $G$, or a single vertex.
Without loss of generality, suppose that $u\in H_0$ and $v\in H_k$.

Note that the number of edges from $H_i$ to $H_j$ is either 0, 1 or
2 for all $i\neq j$, by construction.
Define the \emph{node-link} multigraph $\widetilde{G}$ of $G$
by replacing each $H_j$ by a single vertex $h_j$, which we call a
\emph{node}, and replacing each edge from a vertex of $H_i$ to a vertex
of $H_j$ by an edge from $h_i$ to $h_j$, which we call a \emph{link}.
In particular, if there is a 2-edge-cut from $H_i$ to $H_j$ in $G$
then the link $h_ih_j$ has multiplicity 2 in $\widetilde{G}$.
Each node has even degree in $\widetilde{G}$, by Lemma~\ref{lem:even-d}(i).
Furthermore, every link in $\widetilde{G}$ belongs to a cycle (possibly
a 2-cycle, which is a double link), by Lemma~\ref{lem:even-d}(ii),
and these cycles in $\widetilde{G}$ must be edge-disjoint,
or the corresponding edge of $G$ cannot be part of a 2-edge-cut.
Therefore $\widetilde{G}$ is a planar multigraph (with edge multiplicity
at most two) which has a
tree-like structure, as illustrated in Figure~\ref{fig:two-cuts}:
the black squares represent the nodes of $\widetilde{G}$.
\begin{figure}[ht!]
\begin{center}
\begin{tikzpicture}[scale=0.67,state/.style={rounded rectangle,minimum size=1.8cm,minimum width=2.5cm}]
\draw[fill=black]  (1,2.2) rectangle  node[black] {} ++(0.25,0.25) ;
\draw[fill=black]  (5.5,2.2) rectangle  node[black] {} ++(0.25,0.25) ;
\draw[fill=black]  (4,2.2) rectangle  node[black] {} ++(0.25,0.25) ;
\draw[fill=black]  (-1,1) rectangle  node[black] {} ++(0.25,0.25) ;
\draw[fill=black]  (2.75,1) rectangle  node[black]{} ++(0.25,0.25) ;
\draw[fill=black]  (6.5,1) rectangle  node[black] {} ++(0.25,0.25) ;
\draw[fill=black]  (6.55,3.4) rectangle  node[black] {} ++(0.25,0.25) ;
\draw[fill=black]  (4.6,3.4) rectangle  node[black] {} ++(0.25,0.25) ;
\draw[fill=black]  (1,4.2) rectangle  node[black] {} ++(0.25,0.25) ;
\draw[fill=black]  (-3,1) rectangle  node[black] {} ++(0.25,0.25) ;
\draw[fill=black]  (8.5,1) rectangle  node[black] {} ++(0.25,0.25) ;
\draw[fill=black]  (10,1.95) rectangle  node[black] {} ++(0.25,0.25) ;
\draw[fill=black]  (10,-0.2) rectangle  node[black] {} ++(0.25,0.25) ;
\draw[fill=black]  (11.45,1) rectangle  node[black] {} ++(0.25,0.25) ;
\draw[fill=black]  (1,-0.45) rectangle  node[black] {} ++(0.25,0.25) ;
\draw[fill=black]  (4.75,-0.45) rectangle  node[black] {} ++(0.25,0.25) ;
\draw  (1.1,3.35) circle(1)  (7.6,1) circle(1) (5.7,3.35) circle(1) (-1.9,1) circle(1) ;
\draw (1,1) node[state,draw] {} (4.75,1) node[state,draw] {} (10.1,1) node[state,draw,scale=0.8] {};
\draw (1.1,-0.8) node  {$h_0$};
\draw (12.2,1.1) node  {$h_k$};
\end{tikzpicture}
\end{center}
\caption{The node-link graph $\widetilde{G}$}\label{fig:two-cuts}
\end{figure}
A node of $\widetilde{G}$ which belongs to more than one cycle is
a \emph{join node}.
The join nodes each have degree 4, and
all other nodes of $\widetilde{G}$ have degree two.

Since we seek an upper bound on $\Lambda(u,v)$
for a given number of vertices $n$, we may assume that
all nodes of $\widetilde{G}$ lie on some path from $h_0$ to $h_k$ in
$\widetilde{G}$. (This corresponds to trimming all ``branches'' of
$\widetilde{G}$ which do not lie on the unique ``path'' from $h_0$
to $h_k$ in $\widetilde{G}$, viewed as a tree.)  This only removes
2-edge-cuts of $G$ which do not separate $u$ and $v$.
Hence for an upper bound, by relabelling the nodes if necessary,
we can assume that $\widetilde{G}$ consists of $\ell$ cycles
$C_1,\ldots, C_\ell$,
where
$h_0$ belongs only to $C_1$ and $h_k$ belongs only to $C_\ell$,
such that $C_j$ and $C_{j+1}$ intersect
at a single join node $h_j$, for $j=1,\ldots, \ell-1$,
while all other cycles are disjoint.
This situation is illustrated in Figure~\ref{fig:cuts-path}.
(In particular, $h_0$ and $h_k$ are not join nodes.)
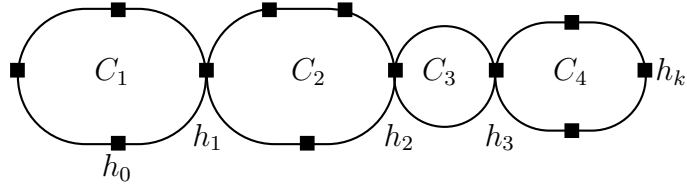
\begin{figure}[ht!]
\begin{center}
\begin{tikzpicture}[scale=0.67,state/.style={rounded rectangle,minimum size=1.8cm,minimum width=2.5cm}]
\draw[fill=black]  (1,2.2) rectangle  node[black] {} ++(0.25,0.25) ;
\draw[fill=black]  (5.5,2.2) rectangle  node[black] {} ++(0.25,0.25) ;
\draw[fill=black]  (4,2.2) rectangle  node[black] {} ++(0.25,0.25) ;
\draw[fill=black]  (-1,1) rectangle  node[black] {} ++(0.25,0.25) ;
\draw[fill=black]  (2.75,1) rectangle  node[black]{} ++(0.25,0.25) ;
\draw[fill=black]  (6.5,1) rectangle  node[black] {} ++(0.25,0.25) ;
\draw[fill=black]  (8.5,1) rectangle  node[black] {} ++(0.25,0.25) ;
\draw[fill=black]  (10,1.95) rectangle  node[black] {} ++(0.25,0.25) ;
\draw[fill=black]  (10,-0.2) rectangle  node[black] {} ++(0.25,0.25) ;
\draw[fill=black]  (11.45,1) rectangle  node[black] {} ++(0.25,0.25) ;
\draw[fill=black]  (1,-0.45) rectangle  node[black] {} ++(0.25,0.25) ;
\draw[fill=black]  (4.75,-0.45) rectangle  node[black] {} ++(0.25,0.25) ;
\draw (7.6,1) circle(1); 
\draw (1,1) node[state,draw] {} (4.75,1) node[state,draw] {} (10.1,1) node[state,draw,scale=0.8] {};
\draw (12.1,1.1) node  {$h_k$};
\draw (1.1,-0.8) node  {$h_0$};
\draw (2.9,-0.2) node  {$h_1$};
\draw (6.7,-0.2) node  {$h_2$};
\draw (8.7,-0.2) node  {$h_3$};
\draw (1,1.1) node  {$C_1$} (4.9,1.1) node  {$C_2$} (7.5,1.1) node  {$C_3$}
(10.1,1.1) node  {$C_4$};
\end{tikzpicture}
\end{center}
\caption{The node-link graph after trimming: an example with $\ell=4$ cycles. }\label{fig:cuts-path}
\end{figure}

Each cycle $C_j$ can be considered as the disjoint union of
two paths between $h_{j-1}$ and $h_j$,
or between $h_{\ell-1}$ and $h_k$ if $j=\ell$.
Denote by $a_j$, $a_j'$ the lengths of these
two paths around $C_j$.
Any 2-edge-cut which separates $u$ from $v$
in $G$ corresponds to a pair of links in $\widetilde{G}$ which both
belong to some cycle $C_j$, with one link from each of the two paths
which comprise $C_j$.  Hence
$\Lambda(u,v) = \sum_{j=1}^\ell a_j a_j'$.
If the length $a_j+a_j'$ of each cycle is fixed then
the quantity $a_j a'_j$ is maximised when $|a_j-a_j'|\leq 1$.
Therefore
$\Lambda(u,v)\leq \sum_{i=1}^\ell s_i^2$, where $s_j = |C_j|/2\geq 1$.
(This upper bound holds even when some cycles have odd length.)
Now, the total number of nodes in $\widetilde{G}$ is
$N=2\left(\sum_{j=1}^\ell s_j\right) -(\ell-1)$, by inclusion-exclusion.
Thus, for a given number of nodes $N$ and a given value of $\ell$,
the bound $S$ is maximised by setting $s_j=1$ for $j=2,\ldots,\ell$, and
$s_1 = (N - \ell+1)/2$.   Therefore
\begin{equation}
\label{Tdef}
 \Lambda(u,v)\leq T = \frac{(N - \ell + 1)^2}{4} + \ell-1.
\end{equation}

First suppose that $r\geq 3$. Then every subgraph $H_j$
must contain a vertex
which is only adjacent to other vertices of $H_j$, since every node in
$\widetilde{G}$ has degree at most 4. Therefore each $H_j$ must contain
at least
$2r+1$ vertices, so there are $N\leq n/(2r+1)$ nodes in $\widetilde{G}$.
Note that $\deriv{T}{\ell} = -\ell + 1$ which equals zero when $\ell=1$
and is negative when $\ell >1$. So $T$ is largest when
$\ell=1$, giving
$\Lambda(u,v)\leq T\leq N^2/4 \leq n^2/(4r+2)^2$,
as claimed.
An extremal graph $G$ is shown in Figure~\ref{fig:maxnumber}: each
grey block is $K_{2r+1}$ minus an edge. This construction gives an
extremal example for any positive $n=0\bmod (4r+2)$, so this upper
bound is asymptotically tight.
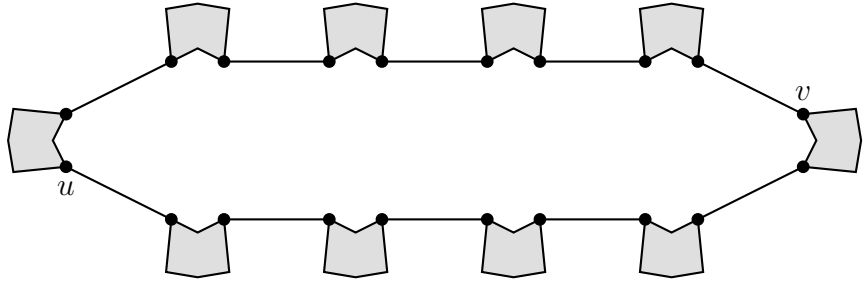
\begin{figure}[ht!]
\begin{center}
\begin{tikzpicture}[scale=0.7]
\begin{scope}[rotate=90]
\draw [fill=lightgray, fill opacity=0.5] (0,0)--(-0.1,1)--(0.5,1.1)--(1.1,1)--(1,0)--(0.5,0.25)--cycle;
\draw [fill] (0,0) circle (0.1) node (a1) {} (-.4,0) node {$u$} ;
\draw [fill] (1,0) circle (0.1) node (a2) {};
\end{scope}
\begin{scope}[xshift=2cm,yshift=2cm]
\draw [fill=lightgray, fill opacity=0.5] (0,0)--(-0.1,1)--(0.5,1.1)--(1.1,1)--(1,0)--(0.5,0.25)--cycle;
\draw [fill] (0,0) circle (0.1) node (a3) {}  ;
\draw [fill] (1,0) circle (0.1) node (a4) {};
\end{scope}
\begin{scope}[xshift=5cm,yshift=2cm]
\draw [fill=lightgray, fill opacity=0.5] (0,0)--(-0.1,1)--(0.5,1.1)--(1.1,1)--(1,0)--(0.5,0.25)--cycle;
\draw [fill] (0,0) circle (0.1) node (a5) {}  ;
\draw [fill] (1,0) circle (0.1) node (a6) {};
\end{scope}
\begin{scope}[xshift=8cm,yshift=2cm]
\draw [fill=lightgray, fill opacity=0.5] (0,0)--(-0.1,1)--(0.5,1.1)--(1.1,1)--(1,0)--(0.5,0.25)--cycle;
\draw [fill] (0,0) circle (0.1) node (a7) {}  ;
\draw [fill] (1,0) circle (0.1) node (a8) {};
\end{scope}
\begin{scope}[xshift=11cm,yshift=2cm]
\draw [fill=lightgray, fill opacity=0.5] (0,0)--(-0.1,1)--(0.5,1.1)--(1.1,1)--(1,0)--(0.5,0.25)--cycle;
\draw [fill] (0,0) circle (0.1) node (a9) {}  ;
\draw [fill] (1,0) circle (0.1) node (a10) {};
\end{scope}
\begin{scope}[xshift=14cm,yshift=1cm,rotate=-90]
\draw [fill=lightgray, fill opacity=0.5] (0,0)--(-0.1,1)--(0.5,1.1)--(1.1,1)--(1,0)--(0.5,0.25)--cycle;
\draw [fill] (0,0) circle (0.1) node (b10) {} (-.4,0) node {$v$} ;
\draw [fill] (1,0) circle (0.1) node (b9) {};
\end{scope}
\begin{scope}[xshift=3cm,yshift=-1cm,rotate=180]
\draw [fill=lightgray, fill opacity=0.5] (0,0)--(-0.1,1)--(0.5,1.1)--(1.1,1)--(1,0)--(0.5,0.25)--cycle;
\draw [fill] (0,0) circle (0.1) node (b8) {}  ;
\draw [fill] (1,0) circle (0.1) node (b7) {};
\end{scope}
\begin{scope}[xshift=6cm,yshift=-1cm,rotate=180]
\draw [fill=lightgray, fill opacity=0.5] (0,0)--(-0.1,1)--(0.5,1.1)--(1.1,1)--(1,0)--(0.5,0.25)--cycle;
\draw [fill] (0,0) circle (0.1) node (b6) {}  ;
\draw [fill] (1,0) circle (0.1) node (b5) {};
\end{scope}
\begin{scope}[xshift=9cm,yshift=-1cm,rotate=180]
\draw [fill=lightgray, fill opacity=0.5] (0,0)--(-0.1,1)--(0.5,1.1)--(1.1,1)--(1,0)--(0.5,0.25)--cycle;
\draw [fill] (0,0) circle (0.1) node (b4) {}  ;
\draw [fill] (1,0) circle (0.1) node (b3) {};
\end{scope}
\begin{scope}[xshift=12cm,yshift=-1cm,rotate=180]
\draw [fill=lightgray, fill opacity=0.5] (0,0)--(-0.1,1)--(0.5,1.1)--(1.1,1)--(1,0)--(0.5,0.25)--cycle;
\draw [fill] (0,0) circle (0.1) node (b2) {}  ;
\draw [fill] (1,0) circle (0.1) node (b1) {};
\end{scope}
\draw (a2)--(a3) (a4)--(a5) (a6)--(a7) (a8)--(a9) (a10)--(b10)
(b9)--(b2) (b1)--(b4)  (b3)--(b6) (b5)--(b8) (b7)--(a1) ;
\end{tikzpicture}
\end{center}
\caption{Extremal graph when $r\geq 3$.}\label{fig:maxnumber}
\end{figure}
Certainly the weaker bound $\Lambda(u,v)\leq n^2/(15 r^2)$ also holds
when $r\geq 3$.

Now suppose that $r=2$ and consider the induced subgraphs $H_0,\ldots, H_k$.
Again we may assume that $s_2 = \cdots = s_\ell = 1$, giving (\ref{Tdef}).
If $h_j$ is a join node of $\widetilde{G}$ then $H_j$ may be as small as
a single vertex, since each join node of $\widetilde{G}$ has degree 4.
We call such a join node $h_j$ a \emph{singleton}.
Otherwise, if $H_j$ contains a vertex whose neighbourhood
is contained within $H_j$ then $H_j$ contains at least $2r+1=5$ vertices.
In particular, this must hold for $h_k$ and for all nodes in $C_1$ except $h_1$.
A final option (only available when $r=2$) is that a join node $h_j$
corresponds to to a subgraph $H_j$ which is isomorphic to $K_4$,
and hence has precisely four vertices: two of these vertices are joined
to a vertex in $H_{j-1}$ and two are joined to $H_{j+1}$ (or $H_k$,
if $j=\ell-1$).   Here we use the fact that any part $U_j$ of
the partition of $V$ with more than one vertex must have at least 4 vertices,
as $G[U_j]$ must be 3-edge-connected.
Since $G$ has no repeated vertices, at most every second join node can
be a singleton, with the remaining join nodes contributing at least 4 vertices
to $G$.

It follows that $\Lambda(u,v) \leq T = s_1^2 + \ell-1$  and
\[ n \geq 10s_1 + \left\lceil\frac{\ell-1}{2}\right\rceil
         + 4\left\lceil\frac{\ell-1}{2}\right\rceil
  = \begin{cases} 10s_1 + 5(\ell-1)/2 & \text{ if $\ell$ is odd,}\\
     10s_1 + 5\ell/2 - r & \text{ if $\ell$ is even.}
\end{cases}
\]
For an upper bound on $\Lambda(u,v)$ we take equality the above expression
for $n$ and assume that $\ell$ is even, leading to
$\Lambda(u,v) \leq T = s_1^2  - 4 s_1 + 2(n+4)/5$.
(If $\ell$ is odd then only the term which is independent of $s_1$ changes.)
Hence $\deriv{T}{s_1}$ is negative when $s_1=1$, zero when $s_1=2$
and positive for $s_1\geq 3$. It follows that for a given value of $n$,
the maximum possible value of $T$ occurs when $s_1=1$ or when $s_1$ is
as large as possible.

When $s_1=1$ we have $\Lambda(u,v)\leq \ell\leq 2(n-6)/5$, at least
when $\ell$ is even.
A graph $G$ on $n$ vertices which meets this bound
can be constructed whenever $n\geq 11$ and $n=1\bmod 5$.
The example shown in Figure~\ref{fig:n=21} has $n=21$ vertices
and $\ell=6$ 2-edge-cuts between $u$ and $v$.
\begin{figure}[ht!]
\begin{center}
\begin{tikzpicture}[scale=0.6]
\draw [fill] (-1,1) circle (0.1) node (a0) {};
\draw [fill] (0,0) circle (0.1) node (a1) {} (-1.4,1) node {$u$} ;
\draw [fill] (2,0) circle (0.1) node (a2) {};
\draw [fill] (2,2) circle (0.1) node (a3) {};
\draw [fill] (0,2) circle (0.1) node (a4) {};
\draw [fill] (4,1) circle (0.1) node (a5) {};
\draw (a4)--(a0)--(a1) (a2)--(a0)--(a3) ;
\draw (a1)--(a2) (a3)--(a4)--(a1)--(a3) (a2)--(a4) (a2)--(a5)--(a3);
\draw [fill] (6,0) circle (0.1) node (b1) {};
\draw [fill] (8,0) circle (0.1) node (b2) {};
\draw [fill] (8,2) circle (0.1) node (b3) {};
\draw [fill] (6,2) circle (0.1) node (b4) {};
\draw [fill] (10,1) circle (0.1) node (b5) {};
\draw (b1)--(b2)--(b3)--(b4)--(b1)--(b3) (b2)--(b4) (b2)--(b5)--(b3) (b1)--(a5)--(b4);
\draw [fill] (12,0) circle (0.1) node (c1) {};
\draw [fill] (14,0) circle (0.1) node (c2) {};
\draw [fill] (14,2) circle (0.1) node (c3) {};
\draw [fill] (12,2) circle (0.1) node (c4) {};
\draw [fill] (16,1) circle (0.1) node (c5) {};
\draw (c1)--(c2)--(c3)--(c4)--(c1)--(c3) (c2)--(c4) (c2)--(c5)--(c3) (c1)--(b5)--(c4);
\draw [fill] (18,0) circle (0.1) node (d1) {};
\draw [fill] (20,0) circle (0.1) node (d2) {};
\draw [fill] (20,2) circle (0.1) node (d3) {};
\draw [fill] (18,2) circle (0.1) node (d4) {} (21.4,1) node {$v$} ;
\draw (d1)--(d2)--(d3)--(d4) (d1)--(d3) (d2)--(d4) (d1)--(c5)--(d4);
\draw [fill] (21,1) circle (0.1) node (d0) {};
\draw (d4)--(d0)--(d1) (d2)--(d0)--(d3) ;
\end{tikzpicture}
\end{center}
\caption{Maximum number of separating 2-edge-cuts for $n=21$}\label{fig:n=21}
\end{figure}
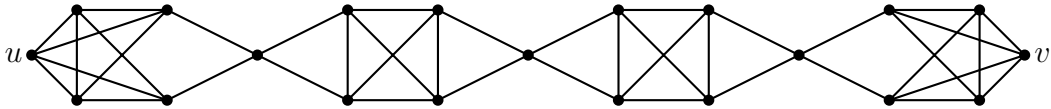
Since $2(n-6)/5\leq n^2/60$ for all values of $n$, this establishes the
required bound.

Next, consider the subcase when $s_1$ is as large as possible, for
a fixed value of $n$. This is achieved by setting $\ell=1$, so
there are no join nodes. Then every $H_j$ has
at least 5 vertices, so for an upper bound we take
$s_1 = n/10$ and $\Lambda(u,v)\leq s_1^2 = n^2/100 < n^2/60$.
This completes the proof.
\end{proof}

\bibliographystyle{plain}

\end{document}